\documentclass[12pt]{article}
\usepackage[T1]{fontenc}
\usepackage[latin9]{inputenc}
\usepackage{color}
\usepackage{prettyref}
\usepackage{float}
\usepackage{amsthm}
\usepackage{amsmath}
\usepackage{amssymb}
\usepackage[authoryear]{natbib}
\usepackage[unicode=true,pdfusetitle,
 bookmarks=true,bookmarksnumbered=true,bookmarksopen=true,bookmarksopenlevel=3,
 breaklinks=false,pdfborder={0 0 1},backref=false,colorlinks=false]
 {hyperref}

\makeatletter

\providecommand{\tabularnewline}{\\}

\newenvironment{lyxcode}
{\par\begin{list}{}{
\setlength{\rightmargin}{\leftmargin}
\setlength{\listparindent}{0pt}% needed for AMS classes
\raggedright
\setlength{\itemsep}{0pt}
\setlength{\parsep}{0pt}
\normalfont\ttfamily}%
 \item[]}
{\end{list}}
 \theoremstyle{definition}
 \newtheorem*{defn*}{\protect\definitionname}
\theoremstyle{plain}
\newtheorem{thm}{\protect\theoremname}

\usepackage[all]{xy}
\usepackage{bbm}
\usepackage{yfonts}
\usepackage{amsfonts}
\usepackage{MnSymbol}
\usepackage{rotating}

\usepackage[samesize]{cancel}
\numberwithin{equation}{section}
\bibpunct{(}{)}{;}{a}{,}{,}

\newcommand{\xyR}[1]{% 
\xydef@\xymatrixrowsep@{#1}}

\newcommand{\xyC}[1]{% 
\xydef@\xymatrixcolsep@{#1}}

\makeatother

  \providecommand{\definitionname}{Definition}
\providecommand{\theoremname}{Theorem}

\begin{document}

\title{Multideviations: The hidden structure of Bell's theorems}

\author{Brandon Fogel}
\maketitle
\begin{abstract}
Specification of the strongest possible Bell inequalities for arbitrarily
complicated physical scenarios---any number of observers choosing
between any number of observables with any number of possible outcomes---is
currently an open problem. Here I provide a new set of tools, which
I refer to as ``multideviations'', for finding and analyzing these
inequalities for the fully general case. In Part I, I introduce the
multideviation framework and then use it to prove an important theorem:
the Bell distributions can be generated from the set of joint distributions
over all observables by deeming specific degrees of freedom unobservable.
In Part II, I show how the theorem provides a new method for finding
tight Bell inequalities. I then specify a set of new tight Bell inequalities
for arbitrary event spaces---the ``even/odd'' inequalities---which
have a straightforward interpretation when expressed in terms of multideviations.
The even/odd inequalities concern degrees of freedom that are independent
of those involved in parameter independence, raising the possibility
of a new Bell's theorem with stronger philosophical implications.
Also, contrary to expectations, the violation of the inequalities
by quantum mechanics increases in size with the number of systems.\end{abstract}
\begin{lyxcode}
\global\long\def\bfempty{OOPS}

\global\long\def\none{}

\global\long\def\bfIF#1#2#3#4{\ifx#1#2{#3}\else{#4}\fi}

\global\long\def\function#1#2{#1\bfIF{#2}{\none}{}{\left(#2\right)}}

\global\long\def\itp#1{\widetilde{#1}}
\global\long\def\itpx{\itp x}
\global\long\def\itpxfull{\itp x_{\itp p}}
\global\long\def\itpy{\itp y}
\global\long\def\itpp{\itp p}
\global\long\def\itpq{\itp q}
\global\long\def\itpgamma{\itp{\gamma}}
\global\long\def\itpalpha{\itp{\alpha}}
\global\long\def\itpeta{\itp{\eta}}
\global\long\def\cminus{\circleddash}
\global\long\def\seq{\subseteq}
\global\long\def\xor{\;\veebar\;}

\global\long\def\inter{\cap}
\global\long\def\un{\cup}

\global\long\def\dsR{\mathbb{R}}

\global\long\def\scL{\mathcal{L}}
\global\long\def\scP{\mathcal{P}}
\global\long\def\scH{\mathcal{H}}
\global\long\def\dsTwo{\mathbbm{2}}
\global\long\def\dsOne{\mathbbm{1}}
\global\long\def\pioneerset{\mathfrak{X}}
\global\long\def\scS{\mathcal{S}}

\global\long\def\SampleFS{\Pi N_{V}}
\global\long\def\SampleES{\left(V,M,N\right)}
\global\long\def\unionM{\cup M}
\global\long\def\SampleESFull{\left(V,M_{V},N_{\cup M_{V}}\right)}
\global\long\def\SampleJMS{\Pi M_{V}}
\global\long\def\SampleJOS{\Pi N_{\itpp}}
\global\long\def\SampleOmni{\Pi N_{\unionM}}
\global\long\def\SampleOmniU{SOU}
\global\long\def\sos{--sos--}

\end{lyxcode}
\global\long\def\lattice#1{\function{\scL}{#1}}

\begin{lyxcode}
\global\long\def\multideviation#1#2#3#4{\function{Q_{#2}^{#1\bfIF{#4}{\none}{}{,#4}}}{#3}}

\global\long\def\MDvector#1#2#3{\function{\vec{q}_{#2}^{\,#1}}{#3}}

\global\long\def\lmultideviation#1#2#3#4{\function{W_{#2}^{#1\bfIF{#4}{\none}{}{,#4}}}{#3}}

\global\long\def\msf#1#2#3{\function{q_{#2}^{#1}}{#3}}
\global\long\def\lmsf#1#2#3{\function{w_{#2}^{#1}}{#3}}

\end{lyxcode}
\global\long\def\simplemult#1{\multideviation{#1}{\none}{\none}{\none}}
\global\long\def\nnsimplemult#1#2{\overline{Q}_{#2}^{#1}}

\begin{lyxcode}
\global\long\def\mcprobability#1#2#3{\function{P_{#2}^{#1}}{#3}}

\global\long\def\powerset#1{\function{\scP}{#1}}
\global\long\def\singletonset#1{\function{\scS}{#1}}

\global\long\def\ddsymbol{\mathfrak{K}}

\global\long\def\discretediff#1#2#3{\function{\ddsymbol}{#1\bfIF{#2}{\none}{}{,#2},#3}}

\global\long\def\bra#1{\left<#1\right|}
\global\long\def\ket#1{\left|#1\right>}
\global\long\def\braket#1#2#3{\bra{#1}#2\ket{#3}}

\end{lyxcode}

\section{Introduction}

In its original form, Bell's theorem described a rather simple physical
scenario: two observers each choosing between two possible measurements
with two possible outcomes each (see \citealt{Bell64,Bell1966}).
The derived empirical limits, the Bell inequalities, were eventually
given a complete description (see \citealt{ClauserHorne1974,Fine1982}).
Subsequent attempts to generalize the theorem to more complicated
scenarios have yielded some notable results, but a similarly complete
and systematic treatment has not yet been achieved.%
\footnote{Notable early attempts include \citet{Svetlichny:1987ui}, \citet{Mermin:1990kc},
and \citet{GHZS:1990}, which provide derivations of empirical limitations
for systems wth three or four particles. \citet{Peres99} provides
a general method for deriving inequalities for higher-dimensional
systems, although few inequalities are actually presented. \citet{PitwoskySvozil2001}
provide a complete list of tight Bell inequalities for three observers
choosing between two binary observables, and \citet{WernerWolf2001}
provide a large class of Bell inequalities, many of which are not
tight (i.e., maximally restrictive). Other notable partial results
include \citet{CollinsGisinEtAl:NBodySeparability:2002}, \citet{ZukowskiBrukner:2002},
\citet{Uffink:2002fk}, and, more recently \citealt{BancalGisin:2011}. %
} The primary obstacle has been the computational complexity of the
problem, which grows exponentially with each of the parameters, particularly
with the number of observers.

The present paper takes a step toward taming that complexity and,
in so doing, exposes some of the deeper structure underlying Bell's
theorems.

The paper is organized into two parts. Part I comprises sections \ref{sec:Conceptual-introduction}-\ref{sec:A-projection-theorem}.
Section \ref{sec:Conceptual-introduction} contains non-technical
presentations of the main results of Part I. Section \ref{sec:Multideviations}
introduces a set of new mathematical tools, which I dub ``multideviations'':
correlation functions that decompose joint probability distributions
into independent degrees of freedom for each subset of observers.
Because the tools have a generality beyond the application to Bell's
theorem, I provide a systematic, if abbreviated, treatment.

Section \ref{sec:Multiple-context-event-spaces} shows how to apply
the multideviation framework to the distributions used in Bell's theorems.
Section \ref{sec:A-projection-theorem} contains an important theorem:
the distributions obeying the Bell inequalities are precisely those
generated by considering specific multideviation degrees of freedom
inaccessible in joint distributions over all observables.

Part II comprises sections \ref{sec:Method-for-finding-BIs}-\ref{sec:Violations in QM}.
Section \ref{sec:Method-for-finding-BIs} uses the new theorem, along
with a bit of matroid theory, a well-established branch of combinatorics,
to outline a new method for finding tight Bell inequalities%
\footnote{A tight Bell inequality is a Bell inequality that is extremal in the
sense that it cannot be written as a linear combination of other Bell
inequalities.\label{fn:Tight-Bell-Inequalities}%
} for arbitrary physical scenarios.

In section \ref{sec:Pioneer-sets}, I use the new method to provide
a large class of solutions for the general case---tight Bell inequalities
for any number of observers each choosing between any number of observables
each with any number of possible outcomes. In section \ref{sec:Intepretation-of-the-even/odd-inequalities},
I provide a convenient conceptual and graphical interpretation of
an important subset of the new inequalities, which I refer to as the
``even/odd inequalities''. In section \ref{sec:Philosophical-importance-of-EO-inqs},
I outline an important feature of these inequalities: they concern
degrees of freedom that are independent of those involved in parameter
independence. This means they could theoretically be derived without
that condition or anything equivalent, thus allowing for a sharper
philosophical conclusion than is achieved with the full complement
of Bell inequalities.

In section \ref{sec:Violations in QM}, I show that quantum mechanics
violates a subset of these inequalities with particularly simple form.
One somewhat surprising result is that the size of the violation increases
with the number of observers and rather quickly converges to the theoretical
maximum. If the violation of these inequalities is taken to represent
a peculiarly non-classical effect, then one might have expected the
effect to diminish as the number of systems is increased, which is
often considered a classical limit; yet this is not the case.

\part{Multideviations and the projection theorem}

In this part, I present the multideviation framework, show how it
is applied to Bell's theorem, and then prove an important projection
theorem concerning distributions satisfying the Bell inequalities.
For the reader uninterested in the full technical presentation, I
offer in section \ref{sec:Conceptual-introduction} conceptual presentations
of the main results. Section \ref{sub:A-geometric-prelude---the 2x2 case}
describes the new mathematical tools, and section \ref{sub:A-conceptual-prelude-the-projection-theorem}
describes how these tools provide a unique and fundamental perspective
on the mathematical origin of the Bell inequalities.

\section{Conceptual presentations\label{sec:Conceptual-introduction}}

\subsection{Multideviations---A geometric prelude\label{sub:A-geometric-prelude---the 2x2 case}}

Consider two observers, $A$ and $B$, each measuring an observable
with two possible outcomes, so that there are four possible joint
outcomes: $\{1_{A}1_{B},$\\
$1_{A}2_{B},2_{A}1_{B},2_{A}2_{B}\}$. The set of possible probability
distributions over these outcomes has three degrees of freedom: one
for each outcome minus one for the constraint that the probabilities
sum to 1. The set can be represented by a tetrahedron (see fig. \ref{fig:Tetrahedron102-1}),
where the vertices and facets have convenient interpretations. 
\begin{figure}
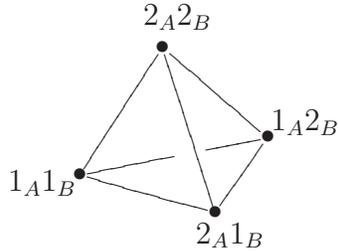

\begin{centering}
\[ 
\xy 
(-10,-5)*{\bullet}="1";
(8,-10)*{\bullet}="2";
(15,0)*{\bullet}="3";
(1,12)*{\bullet}="4";
"1";"2"**\dir{-};
"2";"3"**\dir{-};
"4";"3"**\dir{-};
"1";"4"**\dir{-};
"4";"2"**\dir{-};
{\ar@{-}|>>>>>>>>>>{\hole\hole}"1";"3"};
(-15,-6)*{1_{A}1_{B}};
(10,-13)*{2_{A}1_{B}};
(20,2)*{1_{A}2_{B}};
(3,16)*{2_{A}2_{B}};
\endxy 
\]  
\par\end{centering}

\caption{Geometric representation of a set of probability distributions for
the 2x2 case. Each vertex represents the distribution where that outcome
has probability 1. Each facet represents the distributions where the
opposite vertex has probability 0. All other distributions correspond
to points inside the tetrahedron.}
\label{fig:Tetrahedron102-1} 
\end{figure}

The three degrees of freedom in the set of distributions can be broken
up in a convenient way. One degree of freedom can be used to describe
the probability that $A$ will observe outcome $1_{A}$ or $2_{A}$---the
marginal probabilities for $A$. Only one deegree of freedom is needed
because the probabilities for those outcomes must sum to 1. Another
degree of freedom can be used to describe the marginal probabilities
for $B$. These two degrees of freedom are independent of one another;
they correspond to different subspaces in the vector space containing
the tetrahedron. In fact, the subspaces are orthogonal, and if we
project the tetrahedron into the subspace containing both of them,
it looks like a square (see fig. \ref{fig:Tetrahedron projected}).
\begin{figure}
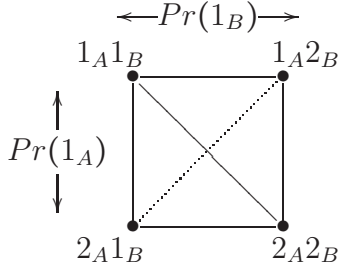

\begin{centering}
\[
\xy 
(10,10)*{\bullet}="12"; 
(10,-10)*{\bullet}="22"; 
(-10,10)*{\bullet}="11"; 
(-10,-10)*{\bullet}="21"; 
(-20,0)*{Pr\!\left(1_A\right)}="PrA"; 
(0,18)*{Pr\!\left(1_B\right)}="PrB"; 
(12,18)*{}="bRightDest"; 
(-12,18)*{}="bLeftDest"; 
(-20,8)*{}="aUpDest"; 
(-20,-8)*{}="aDownDest"; 
(-20,3)*{}="aUpSource"; 
(-20,-3)*{}="aDownSource"; 
{\ar@{->}"aUpSource";"aUpDest"}; {\ar@{->}"aDownSource";"aDownDest"}; 
{\ar@{->}"PrB";"bRightDest"}; 
{\ar@{->}"PrB";"bLeftDest"}; 
"12";"22"**\dir{-}; 
"22";"11"**\dir{-}; 
"21";"11"**\dir{-}; 
"12";"21"**\dir{..}; 
"21";"22"**\dir{-}; 
"12";"11"**\dir{-}; 
(-13,13)*{1_{A}1_{B}}; 
(-13,-13)*{2_{A}1_{B}}; 
(13,13)*{1_{A}2_{B}}; 
(13,-13)*{2_{A}2_{B}}; 
\endxy 
\]
\par\end{centering}

\caption{Projection of the set of 2x2 probability distributions into the marginal
plane. Four of the tetrahedron's edges project to edges in the square,
while two become internal lines (the diagonals). The solid diagonal
line represents the closer of those two edges, and the dotted line
represents the edge on the opposite side of the tetrahedron.}
\label{fig:Tetrahedron projected} 
\end{figure}

The remaining degree of freedom represents information that is not
contained in the marginal probabilities; that is, it concerns only
the correlation between their measurements. 

Although the three degrees of freedom are linearly independent, they
are related to one another by the shape of the tetrahedron. Consider
figure \ref{fig:Tetrahedron projected}, which shows the tetrahedron
from a particular angle. The center of the tetrahedron projects to
the center of the square, between the two edges that project internally
(the diagonal lines). If one starts in the center of the tetrahedron
and then moves out of the page, one will eventually hit the edge connecting
$1_{A}1_{B}$ and $2_{A}2_{B}$. Along this edge, $\function{Pr}{1_{A}}=\function{Pr}{1_{B}}$
and, consequently, $\function{Pr}{2_{A}}=\function{Pr}{2_{B}}$. That
is, the correlation degree of freedom is maximized, and the effect
is that observers $A$ and $B$ will always get the same result. Conversely,
if one starts in the center and moves directly into the page, one
will hit the edge on the other side of the tetrahedron, where the
correlation degree of freedom is minimized and the outcomes are perfectly
anti-correlated (this is the geometric interpretation of equation
\ref{eq:Qmaxmin} in section \ref{sub:Binaries-Inequality-constraints}).

These degrees of freedom correspond to what I have dubbed ``multideviations'':
special functions that decompose joint probability distributions into
linearly independent correlation degrees of freedom. Multidevations
are given systematic treatment in section \ref{sec:Multideviations}.
The multideviation decomposition is of fundamental importance to Bell's
Theorem (see section \ref{sec:A-projection-theorem}), and it provides
a new method for the determination of Bell inequalities for arbitrary
event spaces (see sections \ref{sec:Method-for-finding-BIs} and \ref{sec:Pioneer-sets}).

\subsection{The projection theorem---a conceptual prelude\label{sub:A-conceptual-prelude-the-projection-theorem}}

The projection theorem of section \ref{sec:A-projection-theorem}
says that the distributions described by Bell's theorem can be found
by two different methods. The first is familiar: 1) a set of joint
measurement contexts is formed by allowing each observer to choose
from a set of mutually exclusive observables, 2) a different probability
distribution is specified for each joint measurement context, and
3) certain conditions (often of ontological importance) are imposed
on them. The second is new: 1) a single joint measurement context
is formed by supposing that all observables are measured together,
2) a single probability distribution is specified for that context,
and 3) certain degrees of freedom are deemed unobservable. The projection
theorem shows that the two methods produce the same distributions.

We will consider the two methods schematically for the simplest case:
two observers each choosing between two binary observables. In the
first method, observer $A$ chooses between observables $1$ and $2$,
and observer $B$ chooses between observables $3$ and $4$. There
are thus four measurement contexts: $\left\{ 13,23,14,24\right\} $.
We assign a probability distribution, $P_{ij}$, to each context,
and using the multideviation framework described in the previous section,
we can break each distribution into three degrees of freedom: $\multideviation{\left\{ i,j\right\} }{ij}{\none}{\none},\multideviation{\left\{ i\right\} }{ij}{\none}{\none},\multideviation{\left\{ j\right\} }{ij}{\none}{\none}$.
This is depicted in figure \ref{fig:Degrees-of-freedom-for-mcd-in-2x2x2}.
\begin{figure}
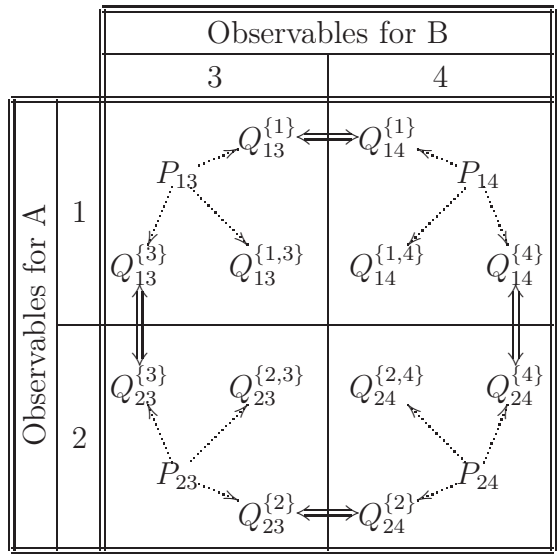

\begin{centering}
\begin{center}
\[
\xy 
(0,0)="O";
(30,0)="X";
(0,30)="Y";
(15,0)="X2";
(0,15)="Y2";
{"O"+"X"};{"O"-"X"}**\dir{-}; 
{"O"+"Y"};{"O"-"Y"}**\dir{-}; 
{"O"+"X"+"Y"};{"O"+"X"-"Y"}**\dir{=}; 
{"O"-"X"-"Y"};{"O"+"X"-"Y"}**\dir{=}; 
{"O"-"X"-"Y"};{"O"-"X"+"Y"}**\dir{=}; 
{"O"+"X"+"Y"};{"O"-"X"+"Y"}**\dir{=}; 
%
% label sizes
(6,0)="LabelX";
(0,6)="LabelY";
(3,0)="LabelX2";
(0,3)="LabelY2";
%short horizontal lines
{"O"-"X"-"Y"-"LabelX"-"LabelX"};{"O"-"X"-"Y"}**\dir{=}; 
{"O"-"X"+"Y"-"LabelX"-"LabelX"};{"O"-"X"+"Y"}**\dir{=}; 
{"O"-"X"-"LabelX"};{"O"-"X"}**\dir{-}; 
%short vertical lines
{"O"-"X"+"Y"+"LabelY"+"LabelY"};{"O"-"X"+"Y"}**\dir{=}; 
{"O"+"X"+"Y"+"LabelY"+"LabelY"};{"O"+"X"+"Y"}**\dir{=}; 
{"O"+"Y"+"LabelY"};{"O"+"Y"}**\dir{-}; 
%long lines
{"O"-"X"-"Y"-"LabelX"};{"O"-"X"+"Y"-"LabelX"}**\dir{-}; 
{"O"-"X"+"Y"+"LabelY"};{"O"+"X"+"Y"+"LabelY"}**\dir{-}; 
{"O"-"X"-"Y"-"LabelX"-"LabelX"};{"O"-"X"+"Y"-"LabelX"-"LabelX"}**\dir{=}; 
{"O"-"X"+"Y"+"LabelY"+"LabelY"};{"O"+"X"+"Y"+"LabelY"+"LabelY"}**\dir{=}; 
%labels
{"O"-"X"+"Y2"-"LabelX2"}*{1};
{"O"-"X"-"Y2"-"LabelX2"}*{2};
{"O"+"Y"-"X2"+"LabelY2"}*{3};
{"O"+"Y"+"X2"+"LabelY2"}*{4};
{"O"+"Y"+"LabelY"+"LabelY2"}*{\textrm{Observables for B}};
{"O"-"X"-"LabelX"-"LabelX2"}*{\begin{sideways}Observables for A\end{sideways}};
%
%Quadrants
%
(20,0)="PX";
(0,20)="PY";
(8,0)="QtopX";
(0,8)="QtopY";
(8,0)="QAX";
(0,25)="QAY";
(25,0)="QBX";
(0,8)="QBY";
{"O"-"QtopX"+"QtopY"}*{Q^{\left\{ 1,3\right\}}_{13}}="Q13"; 
{"O"-"QAX"+"QAY"}*{Q^{\left\{ 1\right\}}_{13}}="Q131"; 
{"O"-"QBX"+"QBY"}*{Q^{\left\{ 3\right\}}_{13}}="Q133"; 
{"O"-"PX"+"PY"}*{P_{13}}="P13"; 
{\ar@{.>}"P13";"Q13"}; 
{\ar@{.>}"P13";"Q131"}; 
{\ar@{.>}"P13";"Q133"}; 
{"O"+"QtopX"+"QtopY"}*{Q^{\left\{ 1,4\right\}}_{14}}="Q14"; 
{"O"+"QAX"+"QAY"}*{Q^{\left\{ 1\right\}}_{14}}="Q141"; 
{"O"+"QBX"+"QBY"}*{Q^{\left\{ 4\right\}}_{14}}="Q144"; 
{"O"+"PX"+"PY"}*{P_{14}}="P14"; 
{\ar@{.>}"P14";"Q14"}; 
{\ar@{.>}"P14";"Q141"}; 
{\ar@{.>}"P14";"Q144"}; 
{"O"-"QtopX"-"QtopY"}*{Q^{\left\{ 2,3\right\}}_{23}}="Q23"; 
{"O"-"QAX"-"QAY"}*{Q^{\left\{ 2\right\}}_{23}}="Q232"; 
{"O"-"QBX"-"QBY"}*{Q^{\left\{ 3\right\}}_{23}}="Q233"; 
{"O"-"PX"-"PY"}*{P_{23}}="P23"; 
{\ar@{.>}"P23";"Q23"}; 
{\ar@{.>}"P23";"Q232"}; 
{\ar@{.>}"P23";"Q233"}; 
{"O"+"QtopX"-"QtopY"}*{Q^{\left\{ 2,4\right\}}_{24}}="Q24"; 
{"O"+"QAX"-"QAY"}*{Q^{\left\{ 2\right\}}_{24}}="Q242"; 
{"O"+"QBX"-"QBY"}*{Q^{\left\{ 4\right\}}_{24}}="Q244"; 
{"O"+"PX"-"PY"}*{P_{24}}="P24"; 
{\ar@{.>}"P24";"Q24"}; 
{\ar@{.>}"P24";"Q242"}; 
{\ar@{.>}"P24";"Q244"}; 
{\ar@{<=>}"Q141";"Q131"}; 
{\ar@{<=>}"Q242";"Q232"}; 
{\ar@{<=>}"Q133";"Q233"}; 
{\ar@{<=>}"Q144";"Q244"}; 
%
%
%(-10,-10)*{\bullet}="test"; 
%{\ar@{->}"test";"origin"}; 
%
\endxy 
\]
\par\end{center}
\par\end{centering}

\caption{\label{fig:Degrees-of-freedom-for-mcd-in-2x2x2}Degrees of freedom
for a multiple-context distribution in the 2x2x2 case. There are four
measurement contexts and thus four probability distributions. Each
distribution is broken into three multideviation degrees of freedom
(dotted arrows), and parameter independence requires multideviations
of the same order to be the same across measurement contexts (double
arrows).}

\end{figure}
At first, there are 12 total degrees of freedom. Parameter independence
requires that the marginals for one observer be independent of the
choice of observable for the other observer. This is enforced by equating
the $Q_{ij}^{\sigma}$ for different $ij$ but the same $\sigma$,
which removes four degrees of freedom (in the figure, this is indicated
by the double arrows). Enforcing determinism or outcome independence
for the underlying states is not represented so easily, but neither
condition reduces the degrees of freedom of the set of observable
distributions (see section \ref{sub:Multiple-context-multideviations-and-conditions}
for more on all of these conditions).

In the second method, we consider all four observables to be measured
by independent observers; there are thus four observers, $\left\{ 1,2,3,4\right\} $,
and a single probability distribution over 16 possible outcomes. There
will thus be 16 different multideviation degrees of freedom (Fig.
\ref{fig:Multideviations-4-observers-all}).
\begin{figure}[H]
\[
\begin{gathered}\begin{array}{c}
\multideviation{\left\{ 1,2,3,4\right\} }{\none}{\none}{\none}\end{array}\\
\begin{array}{cccc}
\multideviation{\left\{ 1,2,3\right\} }{\none}{\none}{\none} & \multideviation{\left\{ 1,2,4\right\} }{\none}{\none}{\none} & \multideviation{\left\{ 1,3,4\right\} }{\none}{\none}{\none} & \multideviation{\left\{ 2,3,4\right\} }{\none}{\none}{\none}\end{array}\\
\begin{array}{cccccc}
\multideviation{\left\{ 1,2\right\} }{\none}{\none}{\none} & \multideviation{\left\{ 1,3\right\} }{\none}{\none}{\none} & \multideviation{\left\{ 1,4\right\} }{\none}{\none}{\none} & \multideviation{\left\{ 2,3\right\} }{\none}{\none}{\none} & \multideviation{\left\{ 2,4\right\} }{\none}{\none}{\none} & \multideviation{\left\{ 3,4\right\} }{\none}{\none}{\none}\end{array}\\
\begin{array}{cccc}
\multideviation{\left\{ 1\right\} }{\none}{\none}{\none} & \multideviation{\left\{ 2\right\} }{\none}{\none}{\none} & \multideviation{\left\{ 3\right\} }{\none}{\none}{\none} & \multideviation{\left\{ 4\right\} }{\none}{\none}{\none}\end{array}\\
\begin{array}{c}
\multideviation{\emptyset}{\none}{\none}{\none}\end{array}
\end{gathered}
\]

\caption{\label{fig:Multideviations-4-observers-all}Multideviation decomposition
for 4 observers.}
\end{figure}
\begin{figure}[H]
\[
\begin{gathered}\begin{array}{c}
\xcancel{\multideviation{\left\{ 1,2,3,4\right\} }{\none}{\none}{\none}}\end{array}\\
\begin{array}{cccc}
\xcancel{\multideviation{\left\{ 1,2,3\right\} }{\none}{\none}{\none}} & \xcancel{\multideviation{\left\{ 1,2,4\right\} }{\none}{\none}{\none}} & \xcancel{\multideviation{\left\{ 1,3,4\right\} }{\none}{\none}{\none}} & \xcancel{\multideviation{\left\{ 2,3,4\right\} }{\none}{\none}{\none}}\end{array}\\
\begin{array}{cccccc}
\xcancel{\multideviation{\left\{ 1,2\right\} }{\none}{\none}{\none}} & \multideviation{\left\{ 1,3\right\} }{\none}{\none}{\none} & \multideviation{\left\{ 1,4\right\} }{\none}{\none}{\none} & \multideviation{\left\{ 2,3\right\} }{\none}{\none}{\none} & \multideviation{\left\{ 2,4\right\} }{\none}{\none}{\none} & \xcancel{\multideviation{\left\{ 3,4\right\} }{\none}{\none}{\none}}\end{array}\\
\begin{array}{cccc}
\multideviation{\left\{ 1\right\} }{\none}{\none}{\none} & \multideviation{\left\{ 2\right\} }{\none}{\none}{\none} & \multideviation{\left\{ 3\right\} }{\none}{\none}{\none} & \multideviation{\left\{ 4\right\} }{\none}{\none}{\none}\end{array}
\end{gathered}
\]

\caption{\label{fig:Multideviations-4-observers-free}Reduced multideviation
decomposition for 4 observers. }
\end{figure}
One degree of freedom, $\multideviation{\emptyset}{\none}{\none}{\none}=\frac{1}{16}$,
is fixed at the outset. The rest can vary between $\pm\multideviation{\emptyset}{\none}{\none}{\none}$
and are linearly independent, although they are related to each other
via some inequality constraints. Each $Q$ measures a correlation
between a subset of observers that cannot be measured by combinations
of the other $Q^{\sigma}$.

Suppose now that some of the correlation degrees of freedom are considered
unobservable---namely, those in which $1$ and $2$ are involved in
a correlation together or similarly with $3$ and $4$. There are
$7$ such degrees of freedom, and, after eliminating those, we are
left with $8$ (see fig. \ref{fig:Multideviations-4-observers-free}).
A careful comparison of the remaining degrees of freedom in figures
\ref{fig:Degrees-of-freedom-for-mcd-in-2x2x2} and \ref{fig:Multideviations-4-observers-free}
shows that there is a 1-1 correspondence.

This correspondence is essentially the content of the projection theorem
of section \ref{sec:A-projection-theorem}: the distributions that
obey all Bell inequalities are generated by ignoring multideviation
correlations involving mutually exclusive observables in joint distributions
on the set of all observables.

\section{Multideviations\label{sec:Multideviations}}

I will now introduce a new framework for analyzing probability distributions;
the degrees of freedom of the distributions will be reconfigured in
terms of novel correlation functions that I have named ``multideviations''.
These do not appear to exist in the current mathematical literature.
Although I have devised this framework specifically for studying the
distributions in Bell's theorem, it is very general and may find application
elsewhere.

\subsection{A note on notation}

The representation of functions of unknown numbers of variables over
products of arbitrary sets can quickly become unwieldy, and the literature
on Bell's theorem has suffered for lack of an efficient and standard
notation. For this reason, I have developed a new notational framework,
the two basic elements of which are the product set and intuple.%
\footnote{See \citet{Fogel2011} for a more systematic treatment.%
} A \emph{product set} is the Cartesian product of an indexed family
of sets:

\begin{equation}
\Pi A_{B}\equiv\prod_{i\in B}A_{i}
\end{equation}
where $A_{B}\equiv\left\{ A_{i}|i\in B\right\} $. The elements of
product sets are \emph{intuples} (``indexed tuples'') and will be
designated by an overhead tilde: $\itpx_{B}\in\Pi A_{B}$. The intuple
components are indicated in straightforward fashion: $x_{i}\in A_{i}$,
where $i\in B$. Intuples can be thought of as ordered tuples or,
if the components carry the indices of their parent sets, as simple
sets. 

Given some product set $\Pi A_{B}$, there is a product set $\Pi A_{\sigma}$
for every $\sigma\subseteq B$. Likewise, $\itpx_{B}$ defines an
intuple $\itpx_{\sigma}$ for every $\sigma\subseteq B$. 

Summation over product sets can be represented compactly:

\begin{equation}
\sum_{\widetilde{x}_{\sigma}}\function f{\itpx_{B}}\equiv\underbrace{\sum_{x_{i}\in A_{i}}\cdots\sum_{x_{k}\in A_{k}}}_{\sigma=\left\{ i,\ldots,k\right\} }\function f{\itpx_{B}}\label{eq:SummationOverFactorSets}
\end{equation}
where $\sigma=\{i,j,\ldots,k\}\subseteq B$.

\subsection{Motivation}

Consider a set of observers $B$. An observer $i\in B$ performs a
single measurement, where the set of possible outcomes is $A_{i}$.
The possible joint outcomes are given by the product set $\Pi A_{B}$,
and the probability of getting the joint outcome $\itpx_{B}$ is given
by the ordinary distribution $\mcprobability{}{}{\itpx_{B}}$. Of
the many different ways to measure the correlation between outcomes
$x_{i}$ and $x_{j}$ for two different observers $i,j\in B$, the
most common is the covariance:
\begin{equation}
\mcprobability{\left\{ i,j\right\} }{}{x_{i}x_{j}}-\mcprobability{\left\{ i\right\} }{}{x_{i}}\mcprobability{\left\{ j\right\} }{}{x_{j}}
\end{equation}
where $\mcprobability{\sigma}{}{\itpx_{\sigma}}\equiv\sum_{B\backslash\sigma}\mcprobability{}{}{\itpx_{B}}$
is a generalized marginal function.

While the covariance has many useful applications, it has a significant
drawback as a pure measure of correlation: it depends on the absolute
value of the marginal functions $\mcprobability{\left\{ i\right\} }{}{x_{i}}$
and $\mcprobability{\left\{ j\right\} }{}{x_{j}}$, not just their
relation to one another. For example, the covariance is 0 when $\mcprobability{\left\{ i\right\} }{}{x_{i}}=\mcprobability{\left\{ j\right\} }{}{x_{j}}=1$
and $\mcprobability{\left\{ i\right\} }{}{x_{i}}=\mcprobability{\left\{ j\right\} }{}{x_{j}}=0$,
even though these are states of seemingly high correlation.

It would thus be desirable to have a measure of correlation that is
independent of the relevant marginal degrees of freedom. Furthermore,
once we have such a measure for the pairs of observers $\left\{ i,j\right\} $,
$\left\{ i,k\right\} $, and $\left\{ j,k\right\} $, we can find
a measure of correlation for the triple $\left\{ i,j,k\right\} $
that is independent of those as well. Repeating this, we can find
an independent correlation function for each $\sigma\subseteq B$.
It turns out that this demand more or less fixes the form of the functions.

\subsection{Multideviation seed functions}

Given an ordinary distribution $\mcprobability{}{}{\itpx_{B}}$ over
a product set $\Pi A_{B}$, we begin with an arbitrary linear combination
of the elements of the distribution:
\begin{equation}
\multideviation{\sigma}{\none}{\itpx_{\sigma}}{\none}\equiv\sum_{\itpy_{B}}\msf{\sigma}{\none}{\itpx_{\sigma},\itpy_{\sigma}}\mcprobability{}{}{\itpy}
\end{equation}
We want the marginals to be written as sums of these functions for
only the relevant degrees of freedom:
\begin{equation}
\mcprobability{\sigma}{}{\itpy_{\sigma}}=c_{\sigma}\sum_{\rho\in\scP\left(B\right)}\multideviation{\rho}{\none}{\itpy_{\rho}}{\none}
\end{equation}
where $c_{\sigma}$ is some constant. It turns out that if we demand
that $\multideviation{\sigma}{\none}{\itpx_{\sigma}}{\none}$ and
$\multideviation{\rho}{\none}{\itpy_{\rho}}{\none}$ be linearly independent
of one another whenever $\sigma\neq\rho$, and that $c_{B}=1$, then
the form of the seed function $q$ is fixed.
\begin{defn*}[Multideviation seed function]
 Given a product set $\Pi A_{B}$ and a cardinality function $n_{\sigma}\equiv\prod_{i\in\sigma}\left|A_{i}\right|$,
where $\sigma\subseteq B$, 
\begin{equation}
\msf{\sigma}B{\itpx_{\sigma},\itpy_{\sigma}}\equiv\frac{1}{n_{B}}\prod_{i\in\sigma}\left(n_{i}\delta_{x_{i}=y_{i}}-1\right)\label{eq:MSF definition}
\end{equation}
 is a \emph{$\sigma$-order multideviation seed function} ($\sigma$-MSF).
\footnote{Using a generalization of the binomial theorem, the MSF can be put
in an alternate form that is frequently very useful: 
\[
\function{q^{\sigma}}{\itpx_{\sigma},\itpy_{\sigma}}=\sum_{\mu\in\scP\left(\sigma\right)}\frac{(-1)^{\left|\sigma\backslash\mu\right|}}{n_{V\backslash\mu}}\delta_{\itpx,\itpy}^{\mu}
\]
 where $\delta_{\itpx,\itpy}^{\mu}\equiv\prod_{i\in\mu}\delta_{x_{i}=y_{i}}$.%
}
\end{defn*}
The MSFs reproduce the Kronecker delta: 
\begin{equation}
\delta_{\itp x=\itp y}=\sum_{\sigma\in\scP\left(B\right)}\msf{\sigma}B{\itpx_{\sigma},\itpy_{\sigma}}\label{MGFCompleteness}
\end{equation}
This means that the MSFs cover the function space; they can be used
to decompose any function over the given product set, $\Pi A_{B}$.

The MSFs have some especially useful algebraic properties. They are
symmetric in the intuple arguments: 
\begin{equation}
\msf{\sigma}B{\itpx_{\sigma},\itpy_{\sigma}}=\msf{\sigma}B{\itpy_{\sigma},\itpx_{\sigma}}
\end{equation}
 They sum to zero for each argument: 
\begin{equation}
\forall i\in\sigma\left[\sum_{x_{i}}\msf{\sigma}B{\itpx_{\sigma},\itpy_{\sigma}}=0\right]\label{MGFSum}
\end{equation}
 Most important, they are closed and orthogonal under a natural inner
product: 
\begin{equation}
\sum_{\itp y_{B}}\msf{\sigma}B{\itpx_{\sigma},\itpy_{\sigma}}\:\msf{\rho}B{\itpy_{\rho},\itp z_{\rho}}=\delta_{\sigma=\rho}\ \msf{\sigma}B{\itpx_{\sigma},\itp z_{\sigma}}\label{MGFOrtho}
\end{equation}

One can think of MSFs as hypermatrices with a particular matrix-multiplication
structure. \eqref{MGFCompleteness} says that the MSFs span the entire
hypermatrix vector space. \eqref{MGFSum} says that there are linear
dependencies among the hypermatrices. \eqref{MGFOrtho} says that
hypermatrices with different $\sigma$ are orthogonal. Thus, the MSFs
can be used to identify a complete set of orthogonal subspaces of
the vector space (or, equivalently, the function space) defined by
the given product set (see appendix \ref{sec:Geometry and multideviations}
for more on vector spaces and multideviations).

\subsection{Multideviations\label{Multideviations}}

\subsubsection{Abitrary functions}

A multideviation is the portion of a function isolated by an MSF:
\begin{defn*}[Multideviation]
 Given a field-valued function $\function f{\itpx_{B}}$ on a factorizable
set $\Pi A_{B}$, a \emph{$\sigma$-multideviation} is given by 
\begin{equation}
\:\multideviation{\sigma}f{\itpx_{\sigma}}B\equiv\sum_{\itp y_{B}}\function f{\itpy_{B}}\ \msf{\sigma}B{\itpx_{\sigma},\itpy_{\sigma}}\label{eq:Multideviation-definition}
\end{equation}

\end{defn*}
Reference to the index set, $B$, can be omitted if it is clear from
the context. Multideviations inherit the algebraic properties of the
MSFs. Thanks to \eqref{MGFCompleteness}, the multideviations decompose
their generating function: 
\begin{equation}
\function f{\itpx_{B}}=\sum_{\sigma\in\scP\left(B\right)}\multideviation{\sigma}f{\itpx_{\sigma}}{\none}
\end{equation}
Since the decomposition is invertible, it is also unique. The summation
property of the MSFs \eqref{MGFSum} means that the multideviations
are not linearly independent within a given $\sigma$: 
\begin{equation}
\forall i\in V\left[\sum_{x_{i}}\multideviation{\sigma}f{\itpx_{\sigma}}{\none}=0\right]\label{MDSum}
\end{equation}
And the inner product property of the MSFs \eqref{MGFOrtho} means
that the multideviations can be picked out by summation with an MSF:
\begin{equation}
\sum_{\itp x_{B}}\,\multideviation{\sigma}f{\itpx_{\sigma}}B\ \msf{\mu}B{\itpx_{\mu},\itpy_{\mu}}=\delta_{\sigma=\mu}\multideviation{\sigma}f{\itpy_{\sigma}}B\label{eq:MD-orthogonality}
\end{equation}
Multideviations provide an alternate representation of a function;
they take the function's degrees of freedom and redistribute them.

\subsubsection{Ordinary probability distributions}

There are several advantages to representing ordinary probability
distributions in terms of multideviations. One concerns the probability
constraint, $\sum_{\itpx_{B}}\mcprobability{}{}{\itpx_{B}}=1$. In
the natural representation, there is not a particularly easy way to
implement this constraint. However, in the multideviation representation,
it becomes 
\begin{equation}
\multideviation{\emptyset}P{\none}{\none}=\frac{1}{n_{B}}
\end{equation}
A single, specific degree of freedom is fixed; the rest are unaffected.
A second advantage of the multideviation representation is that the
first-order marginal degrees of freedom are isolated from the others: 

\begin{equation}
\multideviation{\left\{ i\right\} }P{x_{i}}{\none}=\frac{1}{n_{B}}\left(n_{\left\{ i\right\} }\mcprobability{\left\{ i\right\} }{}{x_{i}}-1\right)
\end{equation}
The first-order marginals and the first-order multideviations are
distinguished only by an offset and a scaling factor; each fixes the
other.

The 2nd-order multideviation, $\multideviation{\left\{ i,j\right\} }P{\itpx_{\left\{ i,j\right\} }}{\none}$,
is what is left of $\mcprobability{\left\{ i,j\right\} }{}{\itpx_{\left\{ i,j\right\} }}$
after the linear dependences on $\multideviation{\left\{ i\right\} }P{x_{i}}{\none}$
and $\multideviation{\left\{ j\right\} }P{x_{j}}{\none}$ have been
removed. In terms of $P$, this is written most economically as
\begin{equation}
\multideviation{\left\{ i,j\right\} }P{\itpx_{\left\{ i,j\right\} }}{\none}=\frac{1}{n_{B}}\left(n_{\left\{ i,j\right\} }\mcprobability{\left\{ i,j\right\} }{}{\itpx_{\left\{ i,j\right\} }}-n_{\left\{ i\right\} }\mcprobability{\left\{ i\right\} }{}{x_{i}}-n_{\left\{ j\right\} }\mcprobability{\left\{ j\right\} }{}{x_{j}}+1\right)
\end{equation}
If one is not convinced by \eqref{MGFOrtho} that $\multideviation{\left\{ i,j\right\} }P{\itpx_{\left\{ i,j\right\} }}{\none}$
and $\multideviation{\left\{ i\right\} }P{x_{i}}{\none}$ are linearly
independent, one need only see how to modify one without the other.%
\footnote{The substitution $\mcprobability{\left\{ i,j\right\} }{}{x_{i}y_{j}}\rightarrow\mcprobability{\left\{ i,j\right\} }{}{x_{i}y_{j}}+\frac{a}{n_{j}}$
for all $y_{j}$ modifies $\multideviation{\left\{ i\right\} }P{x_{i}}{}$
while leaving $\multideviation{\left\{ i,j\right\} }P{\itpx_{\left\{ i,j\right\} }}{}$
unchanged. The substitution $\mcprobability{\left\{ i,j\right\} }{}{\itpy_{\left\{ i,j\right\} }}\rightarrow\mcprobability{\left\{ i,j\right\} }{}{\itpy_{\left\{ i,j\right\} }}+a\left(-1\right)^{\delta_{y_{i}=x_{i}}}\left(-1\right)^{\delta_{y_{j}=x_{j}}}$
for all $y_{i}$ and $y_{j}$ does the opposite.%
}

The higher-order multideviations have a similar form:
\begin{equation}
\multideviation{\sigma}P{\itpx_{\sigma}}{\none}=\frac{1}{n_{B}}\sum_{\rho\in\scP\left(\sigma\right)}\left(-1\right)^{\left|\sigma\backslash\rho\right|}n_{\rho}\mcprobability{\rho}{}{\itpx_{\rho}}
\end{equation}
One of the novelties with multideviations is that one can consider
correlations between different subsets of observers independently.
For example, it is possible for observers 1 and 2 to have outcomes
that are highly correlated while observers 1, 2, and 3 collectively
do not (and vice-versa).

\subsubsection{Inequality constraints}

The probability axiom, $\mcprobability{}{}{\itpx_{B}}\geq0$, causes
multideviations of different orders to be related by inequality constraints: 

\begin{equation}
\sum_{\sigma\in\scP\left(B\right)}\multideviation{\sigma}P{\itpx_{\sigma}}{\none}\geq0\label{eq:Multideviation-basic-inequalities}
\end{equation}
for all $\itpx_{\sigma}\in\Pi A_{B}$. Because the multideviations
are related by \eqref{MDSum}, and because many of them appear in
more than one inequality, the overall structure of these inequalities
can be rather complicated. The binary case (i.e., $\left|A_{i}\right|=2$
for all $i$) is significantly simpler than the general case (see
section \ref{sub:Binaries-Inequality-constraints}).

In general, the inequalities restrict the multideviations to the following
range:

\begin{equation}
-\frac{1}{n_{min}-1}n_{\sigma}^{*}\multideviation{\emptyset}P{\none}{\none}\leq\multideviation{\sigma}P{\itpx_{\sigma}}{\none}\leq n_{\sigma}^{*}\multideviation{\emptyset}P{\none}{\none}
\end{equation}
where $n_{\sigma}^{*}\equiv\prod_{i\in\sigma}\left(\left|A_{i}\right|-1\right)$
and $n_{min}$ is the size of the smallest outcome set.

\subsection{Binary observables and Boolean multideviations\label{sub:Binary-observables-and-Binary-multideviations}}

When the observables are all binary, i.e., when $\left|A_{i}\right|=2$
for all $i\in B$, the multideviations have a particularly convenient
interpretation. This interpretation also applies to more general event
spaces when they are viewed as binary through the use of modified
multideviation functions, which will be referred to as Boolean multideviations.
These turn out to be important elements for the characterization of
the new Bell inequalities introduced in section \ref{sec:Pioneer-sets}.

\subsubsection{Binary observables\label{sub:Even-and-odd}}

When the outcome sets are all binary, there is only one multideviation
degree of freedom per $\sigma$, thanks to \eqref{MDSum}. All multideviations
can be written in terms of an arbitrarily chosen joint outcome:
\begin{equation}
\multideviation{\sigma}P{\itpx_{\sigma}}B=\left(-1\right)^{\left|\sigma\right|-\left|\itpx_{\sigma}\cap\itp 1_{\sigma}\right|}\multideviation{\sigma}P{\itp 1_{\sigma}}B\label{eq:Qx equals Q1}
\end{equation}
We can thus simplify the notation considerably: $\simplemult{\sigma}\equiv\multideviation{\sigma}P{\itp 1_{\sigma}}B$.

In the binary case, we can rewrite the multideviations to give them
a particularly convenient interpretation:
\begin{equation}
\simplemult{\sigma}=\frac{1}{2^{\left|B\right|}}\left(2\function{Pr}{\textrm{even \# of \ensuremath{\sigma}\ outcomes are }2}-1\right)
\end{equation}
where
\begin{equation}
\function{Pr}{\textrm{even \# of \ensuremath{\sigma}\ outcomes are }2}=\sum_{\itpx_{B}}\mcprobability{}{}{\itpx_{B}}\delta_{\left|\itpx_{\sigma}\cap\itp 2_{\sigma}\right|\textrm{is even}}
\end{equation}

When $\sigma=\left\{ i,j\right\} $, 
\begin{equation}
\function{Pr}{\textrm{even \# of }\left\{ i,j\right\} \textrm{ are }2}=\mcprobability{}{}{1_{i}1_{j}}+\mcprobability{}{}{2_{i}2_{j}}
\end{equation}
which is why $\simplemult{\left\{ i,j\right\} }$ measures the extent
to which $x_{i}$ and $x_{j}$ are perfectly correlated.

However, when $\sigma=\left\{ i,j,k\right\} $, 
\begin{equation}
\function{Pr}{\textrm{even in }\itp 2_{\left\{ i,j,k\right\} }}=\mcprobability{}{}{1_{i}1_{j}1_{k}}+\mcprobability{}{}{2_{i}2_{j}1_{k}}+\mcprobability{}{}{2_{i}1_{j}2_{k}}+\mcprobability{}{}{1_{i}2_{j}2_{k}}
\end{equation}
which means $\simplemult{\left\{ i,j,k\right\} }$ does \emph{not
}measure the extent to which $x_{i}$, $x_{j}$, and $x_{k}$ are
perfectly correlated. Rather, it measures a correlation between them
that cannot be gauged by combinations of pairwise correlations among
them---an irreducible fact concerning all three outcomes together.
Likewise for higher orders of $\sigma$.

\subsubsection{Inequality constraints\label{sub:Binaries-Inequality-constraints}}

The inequality constraints deriving from $\mcprobability{}{}{\itpx_{B}}\geq0$
take a simple form when all observables are binary:
\begin{equation}
\sum_{\sigma\in\powerset B}\left(-1\right)^{\left|\sigma\inter\rho\right|}\simplemult{\sigma}\geq0\label{eq:Multideviation-inequalities}
\end{equation}
where $\rho\seq B$.

Of the many consequences of these inequalities, the following are
the most important:

\begin{equation}
-\simplemult{\emptyset}\leq\simplemult{\sigma}\leq\simplemult{\emptyset}
\end{equation}
and
\begin{equation}
\begin{array}{rl}
\simplemult{\sigma}\textrm{ is maximized} & \longrightarrow\forall_{\rho\subseteq B}\left[\simplemult{\rho}=\simplemult{\sigma\ominus\rho}\right]\\
\\
\simplemult{\sigma}\textrm{ is minimized} & \longrightarrow\forall_{\rho\subseteq B}\left[\simplemult{\rho}=-\simplemult{\sigma\ominus\rho}\right]
\end{array}\label{eq:Qmaxmin}
\end{equation}
In particular, note that maximizing $\simplemult{\left\{ i,j\right\} }$
causes $\simplemult{\left\{ i\right\} }=\simplemult{\left\{ j\right\} }$
(perfect correlation), and minimizing $\simplemult{\left\{ i,j\right\} }$
causes $\simplemult{\left\{ i\right\} }=-\simplemult{\left\{ j\right\} }$
(perfect anti-correlation). This is the formal expression of the geometric
relation noted in the discussion of figure \ref{fig:Tetrahedron projected}
in section \ref{sub:A-geometric-prelude---the 2x2 case}.  

While the behavior of the second-order multideviation is straightforward,
that of the higher-orders is more subtle. One can see clearly from
\eqref{eq:Qmaxmin} that maximizing $\simplemult{\left\{ i,j,k\right\} }$
will \emph{not} cause $\simplemult{\left\{ i\right\} }=\simplemult{\left\{ j\right\} }=\simplemult{\left\{ k\right\} }$.
That would be achieved by maximizing $\simplemult{\left\{ i,j\right\} }$,
$\simplemult{\left\{ i,k\right\} }$, and $\simplemult{\left\{ j,k\right\} }$
separately. Instead, $\simplemult{\left\{ i,j,k\right\} }$ measures
how the outcomes are correlated in a way that cannot be measured by
a combination of the lower order multideviations.

The even/odd interpretation provided in the previous section elucidates
this further. For example, if $\simplemult{\sigma}$ is maximized,
then the set of outcomes for observers $\sigma$ must have an even
number of 2's, and thus the extent to which observers $\rho$ and
$\sigma\cminus\rho$ have even numbers of 2's must be correlated.
If one definitely does not, then the other must not either, and vice-versa.

\subsubsection{Boolean multideviations\label{sub:Boolean-multideviations}}

Event spaces where $\left|A_{i}\right|>2$ for at least one $i\in B$
can be viewed as binary through the use of multideviations over lattice
elements, or what I will refer to as \emph{Boolean multideviations}.
The seed functions for these quantities are 
\begin{equation}
\lmsf{\sigma}B{\itpx_{\sigma},\itpalpha_{\sigma}}\equiv\frac{1}{2^{\left|B\right|}}\prod_{i\in\sigma}\left(2\delta_{x_{i}\in\alpha_{i}}-1\right)
\end{equation}
where $\itpalpha_{\sigma}\in\prod_{i\in\sigma}\lattice{A_{i}}$ is
an intuple over the Boolean lattices of the outcome sets.

The Boolean multideviations are
\begin{equation}
\lmultideviation{\sigma}f{\itpalpha_{\sigma}}B\equiv\sum_{\itp y_{B}}\function f{\itpy_{B}}\ \lmsf{\sigma}B{\itpy_{\sigma},\itpalpha_{\sigma}}\label{eq:Lattice-Multideviation-definition}
\end{equation}
Although the derivation is not straightforward, the following inequalties
hold:
\begin{equation}
\sum_{\sigma\in\powerset B}\left(-1\right)^{\left|\sigma\inter\rho\right|}W^{\sigma}\geq0
\end{equation}
for all $\rho\seq B$. Since this is identical in form to \eqref{eq:Multideviation-inequalities},
all of the results in sections \ref{sub:Binaries-Inequality-constraints}
and \ref{sub:Even-and-odd} can be transferred to the general case
by substituting $Q\rightarrow W$, $1_{i}\rightarrow\alpha_{i}$,
and $2_{i}\rightarrow\alpha_{i}^{c}$.

\section{Multiple-context event spaces and distributions\label{sec:Multiple-context-event-spaces}}

\subsection{Definition}

The distributions used in Bell's theorem are not ordinary probability
distributions, but rather collections of ordinary distributions, one
for each possible joint measurement context. The notational framework
of multiple-context event spaces allows efficient characterization
of such distributions for arbitrarily complex physical scenarios.%
\footnote{See \citet{Fogel2011} for a systematic discussion of multiple-context
event spaces and distributions.%
}

A multiple-context event space\emph{ }is an ordered triple $\SampleESFull$
representing
\begin{enumerate}
\item A set of observers: $V=\left\{ A,B,C,\ldots\right\} $,
\item A set of observables for each observer $i\in V$: $M_{i}=\left\{ \alpha_{i},\beta_{i},\ldots\right\} $,
\item A set of outcomes for each observable $p_{i}\in M_{i}$: $N_{p_{i}}=\{1_{p_{i}},2_{p_{i}},\dots\}$.
\end{enumerate}
Note that the indices in each set do not count the elements of the
set, but rather indicate membership in it. For example, when $p_{i}$
iterates over the elements of $M_{i}$, $i$ is held fixed and is
preserved on the variable $p_{i}$ primarily to emphasize that it
represents an element of $M_{i}$. Likewise for $p_{i}$ in the outcome
$1_{p_{i}}$ or outcome variable $x_{p_{i}}$.

Several spaces of physical importance are defined by the event space:
\begin{enumerate}
\item Joint measurement context space : $\Pi M_{V}\equiv\prod_{i\in V}M_{i}$,
\item Joint outcome space: $\SampleJOS\equiv\prod_{i\in V}N_{p_{i}}$,
\item Omni-joint outcome space: $\SampleOmni\equiv\prod_{i\in V}\prod_{p_{i}\in M_{i}}N_{p_{i}}$.
\end{enumerate}

A multiple-context probability distribution is a collection of ordinary
probability distributions, one for each joint measurement context:
\begin{equation}
P:\Pi M_{V}\rightarrow\left\{ f:\Pi N_{\itpq}\rightarrow\left[0,1\right]\ |\ \itp q\in\Pi M_{V}\right\} 
\end{equation}
where 
\begin{equation}
P_{\itp p}:\Pi N_{\itpp}\rightarrow\left[0,1\right]\label{eq:Multiple_Context_Distribution_Constraint}
\end{equation}
and
\begin{align}
\mcprobability{}{\itpp}{\itpx_{\itpp}} & \geq0\label{WPC}\\
\sum_{\itp x_{\itpp}}\mcprobability{}{\itpp}{\itpx_{\itpp}} & =1\label{SPC}
\end{align}

$\mcprobability{}{\itpp}{\itpx_{\itpp}}$ is the probability for getting
outcome $\itpx_{\itpp}$ in joint measurement context $\itp p$.

\subsection{Multiple-context multideviations; motivating conditions\label{sub:Multiple-context-multideviations-and-conditions}}

Since multideviations are defined for ordinary probability distributions,
there will be a separate set of multideviations for each joint measurement
context: $\multideviation{\itpp_{\sigma}}{P_{\itpp}}{\itpx_{\itpp_{\sigma}}}{\itpp}$.
Note that $\itpp$ is being used in several different ways in this
expression: $P_{\itpp}$ identifies the ordinary probability distribution
to be summed over (the $f$ in eq. \ref{eq:Multideviation-definition}),
$\itpp$ on its own is the set of possible observables (the $B$ in
eq. \ref{eq:Multideviation-definition}), and $\itpp_{\sigma}$ is
the set of observables being correlated (the $\sigma$ in eq. \ref{eq:Multideviation-definition}).

The most common motivating conditions used in Bell's theorems have
relatively simple expression in terms of multideviations.

Parameter independence (no-signalling) is a condition on distributions
between different joint measurement contexts:

\begin{equation}
\forall\rho\supseteq\sigma\left[n_{\itpp_{V\backslash\sigma}}\multideviation{\itpp_{\sigma}}{P_{\itpp}}{\itpx_{\itpp_{\sigma}}}{\itpp}=n_{\itpq_{V\backslash\sigma}}\multideviation{\itpp_{\sigma}}{P_{\left(\itpp_{\rho}\itpq_{V\backslash\rho}\right)}}{\itpx_{\itpp_{\sigma}}}{\left(\itpp_{\rho}\itpq_{V\backslash\rho}\right)}\right]
\end{equation}
That is, if two joint measurement contexts share a set of observables,
then the multideviations for that set will be fixed by one another.

Determinism is a condition on ordinary probability distributions.
It requires the probability for one outcome to be $1$ and the rest
to be $0$. For example, suppose $\mcprobability{}{}{\itpx_{B}}=\delta_{\itpx_{B}=\itpy_{B}}$.
Then
\begin{equation}
\multideviation{\sigma}P{\itpx_{B}}B=\msf{\sigma}B{\itpx_{B},\itpy_{B}}
\end{equation}

Finally, outcome independence is also a condition on ordinary distributions:
\begin{equation}
n_{B}\:\multideviation{\sigma}P{\itpx_{\sigma}}B=\prod_{i\in\sigma}\left(n_{B}\:\multideviation{\left\{ i\right\} }P{x_{i}}B\right)
\end{equation}

\subsection{The CHSH inequality\label{sub:CHSH}}

The first clue to the importance of multideviations for Bell's Theorem
comes from the CHSH inequality, the simplest known tight Bell inequality.
The CHSH inequality (\citealt{ClauserHorne1974}) is the strongest
possible Bell-type inequality for a multiple-context distribution
over a 2x2x2 event space ($\left|V\right|=2$, $\left|M_{i}\right|=2$,
and $\left|N_{p_{i}}\right|=2$). In the natural representation, it
has the following form: 
\begin{equation}
\mcprobability{}{p_{i}p_{j}}{x_{p_{i}}x_{p_{j}}}+\mcprobability{}{p_{i}q_{j}}{x_{p_{i}}y_{q_{j}}}+\mcprobability{}{q_{i}p_{j}}{y_{q_{i}}x_{p_{j}}}-\mcprobability{}{q_{i}q_{j}}{y_{p_{i}}y_{q_{i}}}-\mcprobability{\left\{ i\right\} }{p_{i}}{x_{p_{i}}}-\mcprobability{\left\{ j\right\} }{p_{j}}{x_{p_{j}}}\leq0
\end{equation}
Eight distinct inequalities of this form can be generated by choosing
different values for $\itp p_{V}$, $\itp q_{V}$, $\itp x_{\itpp}$,
and $\itp y_{\itpq}$. The inequality is usually expressed in the
terms of arbitrarily chosen expectation values intended to simplify
the appearance.

However, in the multideviation representation, the CHSH inequality
has a natural simplicity:
\begin{equation}
-\frac{1}{2}\le\multideviation{p_{i}p_{j}}{\none}{\none}{\none}+\multideviation{p_{i}q_{j}}{\none}{\none}{\none}+\multideviation{q_{i}p_{j}}{\none}{\none}{\none}-\multideviation{q_{i}q_{j}}{\none}{\none}{\none}\leq\frac{1}{2}\label{eq:CHSH-1-1}
\end{equation}
where $Q^{\itp{\mu}}=\multideviation{\itp{\mu}}{\mcprobability{}{\itp{\mu}}{\none}}{\itp 1_{\itp{\mu}}}{\none}$.
The multideviation representation exposes an interesting fact---the
constraint depends only on the second-order multideviation degrees
of freedom. The appearance of the marginal degrees of freedom in the
natural representation is an artifact of the representation. 

The interpretation of the binary multideviations given in section
\ref{sub:Even-and-odd} allows a useful understanding of the CHSH
inequality. Since the binary multideviations measure how even or odd
the statistics are (relative to a chosen set of outcomes), the inequality
represents a limit on how incompatible the even or odd statistics
for the joint measurement contexts can be with one another. For example,
if we select $\itp 1_{\left\{ 1,2,3,4\right\} }$ as a reference,
then a joint outcome $\itpx_{\itp{\mu}}$ is even or odd depending
on whether an even or odd number of the outcomes are 2. Given the
starting conditions, it is not possible for the joint outcomes to
be odd with certainty in an odd number of the joint measurement contexts
(see fig. \ref{fig:CHSH-even-odd-interpretation}); the CHSH inequalities
express this fact and indicate how close the statistics generated
from those conditions can get to such a state.
\begin{figure}
\begin{centering}
\begin{center}
\[
\xy 
(0,0)="O";
(26,0)="X";
(0,18)="Y";
(13,0)="X2";
(0,9)="Y2";
{"O"+"X"};{"O"-"X"}**\dir{-}; 
{"O"+"Y"};{"O"-"Y"}**\dir{-}; 
{"O"+"X"+"Y"};{"O"+"X"-"Y"}**\dir{=}; 
{"O"-"X"-"Y"};{"O"+"X"-"Y"}**\dir{=}; 
{"O"-"X"-"Y"};{"O"-"X"+"Y"}**\dir{=}; 
{"O"+"X"+"Y"};{"O"-"X"+"Y"}**\dir{=}; 
%
% label sizes
(6,0)="LabelX";
(0,6)="LabelY";
(3,0)="LabelX2";
(0,3)="LabelY2";
%short horizontal lines
{"O"-"X"-"Y"-"LabelX"-"LabelX"};{"O"-"X"-"Y"}**\dir{=}; 
{"O"-"X"+"Y"-"LabelX"-"LabelX"};{"O"-"X"+"Y"}**\dir{=}; 
{"O"-"X"-"LabelX"};{"O"-"X"}**\dir{-}; 
%short vertical lines
{"O"-"X"+"Y"+"LabelY"+"LabelY"};{"O"-"X"+"Y"}**\dir{=}; 
{"O"+"X"+"Y"+"LabelY"+"LabelY"};{"O"+"X"+"Y"}**\dir{=}; 
{"O"+"Y"+"LabelY"};{"O"+"Y"}**\dir{-}; 
%long lines
{"O"-"X"-"Y"-"LabelX"};{"O"-"X"+"Y"-"LabelX"}**\dir{-}; 
{"O"-"X"+"Y"+"LabelY"};{"O"+"X"+"Y"+"LabelY"}**\dir{-}; 
{"O"-"X"-"Y"-"LabelX"-"LabelX"};{"O"-"X"+"Y"-"LabelX"-"LabelX"}**\dir{=}; 
{"O"-"X"+"Y"+"LabelY"+"LabelY"};{"O"+"X"+"Y"+"LabelY"+"LabelY"}**\dir{=}; 
%labels
{"O"-"X"+"Y2"-"LabelX2"}*{p_A};
{"O"-"X"-"Y2"-"LabelX2"}*{q_A};
{"O"+"Y"-"X2"+"LabelY2"}*{p_B};
{"O"+"Y"+"X2"+"LabelY2"}*{q_B};
{"O"+"Y"+"LabelY"+"LabelY2"}*{\textrm{Observables for B}};
{"O"-"X"-"LabelX"-"LabelX2"}*{\begin{sideways}Observables for A\end{sideways}};
%
%Quadrants
%
{"O"-"X2"+"Y2"}*{{\function P{\textrm{Even}}} =1}; 
{"O"+"X2"+"Y2"}*{{\function P{\textrm{Even}}} =1}; 
{"O"-"X2"-"Y2"}*{{\function P{\textrm{Even}}} =1}; 
{"O"+"X2"-"Y2"}*{{\function P{\textrm{Odd}}} =1}; 
%
%{"O"+(-1,1)}*!DRD{\emptyset}; 
%{"O"+(-1,-1)}*!DRU{\left\{A\right\}}; 
%{"O"+(1,1)}*!DLD{\left\{B\right\}}; 
%{"O"+(1,-1)}*!DLU{\left\{A,B\right\}}; 
%
%{"O"+(-2,2)}*!DRD{\textrm{\footnotesize $\emptyset$}}; 
%{"O"+(-1,1)-"Y"}*!DRD{\textrm{\footnotesize $\left\{A\right\}$}}; 
%{"O"+(-1,1)+"X"}*!DRD{\textrm{\footnotesize $\left\{B\right\}$}}; 
%{"O"+(-1,1)+"X"-"Y"}*!DRD{\textrm{\footnotesize $\left\{A,B\right\}$}}; 
%
%Legend
%
{"O"+"X"+"X"+(5,0)}*{\begin{array}{l} \textrm{Even outcomes: 11, 22}\\ \textrm{Odd outcomes: 12, 21} \end{array}}="P13"; 
\endxy 
\]
\par\end{center}
\par\end{centering}

\caption{\label{fig:CHSH-even-odd-interpretation}Interpretation of the CHSH
inequality. The above state is impossible for a system obeying the
upper bound of equation \eqref{eq:CHSH-1-1}. The state in which all
of the above probabilities are zero is disallowed by the lower bound
of \eqref{eq:CHSH-1-1}.}
\end{figure}

In section \ref{sub:Pioneer sets - Special-cases}, the CHSH inequality
is shown to be a special case of a larger classes of inequalities.
The conceptual interpretation offered here is presented systematically
in section \ref{sec:Intepretation-of-the-even/odd-inequalities}.

\section{A projection theorem\label{sec:A-projection-theorem}}

I will now use the multideviation framework to prove an important
result: the Bell distributions---those multiple-context probability
distributions satisfying all Bell inequalities---can be generated
from ordinary probability distributions over the set of omni-joint
outcomes. The existence of this relationship is not surprising in
and of itself; the equivalence of ``existence of the joints'' and
``satisfaction of the Bell inequalities'' is well-known (see fn.
\ref{fn:Existence-of-the-joints-clarification} below). The novel
result is that the mapping is accomplished by ignoring specific multideviation
degrees of freedom, namely, those involving two or more mutually exclusive
observables.\textcolor{magenta}{{} }

\subsection{The theorem\label{sub:The-projection-theorem}}

In the following, the event space is $\SampleESFull$ (see section
\ref{sec:Multiple-context-event-spaces}).

Deterministic distributions are those whose values are all 0 or 1.
Parameter-independent (i.e., no-signalling) distributions are those
whose marginals are independent of the measurement choices of other
observers (see section \ref{sub:Multiple-context-multideviations-and-conditions}).
The deterministic, parameter-independent multiple-context distributions
are thus given by 
\begin{equation}
\function{G_{\itpp}^{\itp{\gamma}}}{\itpx_{\itpp}}=\delta_{\itpx_{\itpp}=\itpgamma_{\itpp}}\label{Vectors in a deterministic, no-signalling polytope-1}
\end{equation}
 where $\itp{\gamma}\in\SampleOmni$. 

The set of distributions that statifsy the Bell inequalities, or Bell
distributions, is the set of convex combinations of deterministic,
parameter-independent distributions: 
\begin{defn*}[Bell distribution]
All Bell distributions can be written as 
\begin{equation}
\mcprobability{}{\itpp}{\itpx_{\itpp}}=\sum_{\itpgamma\in\SampleOmni}\function{\mu}{\itpgamma}\function{G_{\itpp}^{\itp{\gamma}}}{\itpx_{\itpp}}
\end{equation}
 where $\function{\mu}{\itpgamma}\geq0$ and $\sum_{\itpgamma}\function{\mu}{\itpgamma}=1$.\end{defn*}
\begin{thm}
\label{thm:Bell-distributions-are-projections}Bell distributions
are projections of omni-joint distributions, where the ignored degrees
of freedom are multideviations involving two or more mutually-exclusive
observables. \end{thm}
\begin{proof}
First, we express an arbitrary Bell distribution in terms of omni-joint
multideviations:{\allowdisplaybreaks
\begin{align}
\mcprobability{}{\itpp_{V}}{\itpx_{\itpp_{V}}} & =\sum_{\itpgamma_{\cup M}}\function{\mu}{\itpgamma}\function{G_{\itpp}^{\itp{\gamma}}}{\itpx_{\itpp}}\\
 & =\sum_{\itpgamma_{\cup M}}\function{\mu}{\itpgamma}\delta_{\itpx_{\itpp}=\itpgamma_{\itpp}}\\
 & =\sum_{\itpgamma_{\cup M}}\left(\sum_{\rho\in\powerset{\cup M}}\multideviation{\rho}{\mu}{\itpgamma_{\rho}}{\cup M}\right)\delta_{\itpx_{\itpp}=\itpgamma_{\itpp}}\\
 & =\sum_{\itpgamma_{\itpp}}\sum_{\itpgamma_{\cup M\backslash\itpp}}\left(\sum_{\rho\in\powerset{\cup M}}\multideviation{\rho}{\mu}{\itpgamma_{\rho}}{\cup M}\right)\delta_{\itpx_{\itpp}=\itpgamma_{\itpp}}\\
 & =\sum_{\itpgamma_{\itpp}}\delta_{\itpx_{\itpp}=\itpgamma_{\itpp}}\sum_{\rho\in\powerset{\cup M}}\sum_{\itpgamma_{\cup M\backslash\itpp}}\multideviation{\rho}{\mu}{\itpgamma_{\rho}}{\cup M}\\
 & =\sum_{\itpgamma_{\itpp}}\delta_{\itpx_{\itpp}=\itpgamma_{\itpp}}\sum_{\rho\in\powerset{\cup M}}\delta_{\rho\subseteq\itpp}\: n_{\cup M\backslash\rho}\,\multideviation{\rho}{\mu}{\itpgamma_{\rho}}{\cup M}\\
 & =\sum_{\itpgamma_{\itpp}}\delta_{\itpx_{\itpp}=\itpgamma_{\itpp}}\sum_{\rho\in\powerset{\itpp}}n_{\cup M\backslash\rho}\multideviation{\rho}{\mu}{\itpgamma_{\rho}}{\cup M}\\
 & =\sum_{\rho\in\powerset{\itpp}}n_{\cup M\backslash\rho}\multideviation{\rho}{\mu}{\itpx_{\rho}}{\cup M}\\
 & =\sum_{\sigma\in\powerset V}n_{\cup M\backslash\itpp_{\sigma}}\multideviation{\itpp_{\sigma}}{\mu}{\itpx_{\itpp_{\sigma}}}{\cup M}\label{eq:ProjectionTheorem1}
\end{align}
}Next, we express the same distribution in terms of multideviations
over the indicated joint measurement context:
\begin{equation}
\mcprobability{}{\itpp_{V}}{\itpx_{\itpp_{V}}}=\sum_{\sigma\in\powerset V}\multideviation{\itpp_{\sigma}}{P_{\itpp}}{\itpx_{\itpp_{\sigma}}}{\itpp}\label{eq:ProjectionTheorem2}
\end{equation}
By \eqref{eq:MSF definition}, for any $\mu\subseteq V$,
\begin{equation}
\frac{n_{\itpp}}{n_{\cup M}}\msf{\itpp_{\rho}}{\itpp}{\itpx_{\itpp_{\rho}},\itpy_{\itpp_{\rho}}}=\msf{\itpp_{\rho}}{\cup M}{\itpx_{\itpp_{\rho}},\itpy_{\itpp_{\rho}}}
\end{equation}
This means that, for any function $\function f{\itpx_{\itpp}}$,
\begin{equation}
\sum_{\itpx_{\itpp}}\msf{\itpp_{\rho}}{\itpp}{\itpx_{\itpp_{\rho}},\itpy_{\itpp_{\rho}}}\function f{\itpx_{\itpp}}=\sum_{\itpx_{\cup M}}\msf{\itpp_{\rho}}{\cup M}{\itpx_{\itpp_{\rho}},\itpy_{\itpp_{\rho}}}\function f{\itpx_{\itpp}}\label{eq:ProjectionTheorem3}
\end{equation}
Using \eqref{eq:MD-orthogonality} and \eqref{eq:ProjectionTheorem3}
to manipulate \eqref{eq:ProjectionTheorem1} and \eqref{eq:ProjectionTheorem2},
we derive: 
\begin{equation}
\multideviation{\itpp_{\sigma}}{P_{\itpp}}{\itpx_{\itpp_{\sigma}}}{\itpp}=n_{\cup M\backslash\itpp_{\sigma}}\multideviation{\itpp_{\sigma}}{\mu}{\itpx_{\itpp_{\sigma}}}{\cup M}\label{eq:Mapping-between-OJ-and-MC}
\end{equation}
That is, given an omni-joint distribution $\function{\mu}{\itpgamma}$,
we can construct a multiple-context distribution that satisfies the
Bell inequalities using only those multideviation degrees of freedom
that are subsets of the joint measurement contexts, $\itpp\in\SampleJMS$.
These are just the subsets of $\cup M$ that include no more than
one element of $M_{i}$ for each observer $i\in V$.
\end{proof}

\subsection{Geometric interpretation}

\prettyref{thm:Bell-distributions-are-projections} has a straightforward
geometric interpretation. The set of omni-joint distributions corresponds
to a simplex in $\dsR^{\left(n_{\unionM}-1\right)}$. The multideviations
identify orthogonal subspaces of that vector space. The theorem says
that the set of Bell distributions corresponds to a polytope formed
by projecting the simplex into the subspace generated by the set of
multideviation vectors associated with those sets of observables that
can be measured simultaneously (an affine transformation is also needed).
See Appendix \ref{sec:Geometry and multideviations} for more on the
geometric approach.

\subsection{Philosophical consequences}

Let 
\begin{equation}
\Psi\equiv\left\{ \rho\subseteq\cup M\;|\;\forall i\in V\left[\left|M_{i}\cap\rho\right|\leq1\right]\right\} 
\end{equation}
be the collection of all sets of comeasurable observables (i.e., no
two are mutually exclusive). 

Now imagine a world in which all observables $\cup M$ are measured
together, producing a probability distribution, $\mu$, over the omni-joint
outcomes. Imagine, however, that the observers are not permitted to
share their results with one another; rather, some administrator takes
their results, calculates the multideviations, and returns only those
corresponding to elements of $\Psi$. The researchers will be able
to recover some probability distributions corresponding to various
joint measurements, but they will not be able to reconstruct $\mu$
in total. 

\prettyref{thm:Bell-distributions-are-projections} tells us that
this is effectively the situation with Bell distributions. Each Bell
distribution is equivalent to at least one probability distribution
over all observables taken together, but where multideviation correlations
involving mutually exclusive observables are considered inaccessible.

As noted above, this result is an extension of the well-known ``existence
of the joints'' theorem---if a multiple-context distribution satisfies
the Bell inequalities, then there is a joint distribution over all
observables (an ``omni-joint'' distribution) that reproduces the
original distribution as marginals.%
\footnote{\label{fn:Existence-of-the-joints-clarification}\citet{Fine1982}
proved this for the simplest case; \citet{Fogel2011} proved this
for the general case. The importance of this theorem has been critiqued
by \citet{SvetlichnyEtAl1988}, \citet{Butterfield92}, and \citet{MullerPlacek:2001Synthese}.%
} The novelty here is that the multideviations make clear precisely
which aspects of the omni-joint distribution are hidden; or, put another
way, the multideviations show us what information needs to be restored
in order to reconstruct the omni-joint distribution. The inequalities
\eqref{eq:Multideviation-basic-inequalities} determine the ranges
of allowed values for the hidden multideviations. It is when, and
only when, these inequalities are inconsistent, given the observable
multideviation degrees of freedom, that at least one Bell inequality
is violated.

Thus, the multidevations for the elements in $\unionM\backslash\Psi$
 are, effectively, the hidden variables compatible with the Bell inequalities.
A theory may have a richer set of hidden variables, but they will
have to reduce to the multideviations involving mutually exclusive
correlations, if the observable statistics are the Bell distributions.

\part{Tight Bell inequalities}

In this part, I use the projection theorem of section \ref{sec:A-projection-theorem},
along with matroid theory, to outline a new method for finding tight
Bell inequalities. I then present a new class of such inequalities,
which turn out to have a straightforward interpretation, and show
that they are violated by quantum mechanics.

\section{Method for finding BIs\label{sec:Method-for-finding-BIs}}

\subsection{Preliminaries}

A Bell inequality is an inequality satisfied by a multiple-context
distribution subject to certain constraints. Some arbitrariness exists
in the choice of constraints, since different choices produce the
same set of distributions (and hence the same inequalities). For convenience,
I have chosen to work with parameter independence (i.e., no-signalling)
and determinism.%
\footnote{By taking the observable distributions to be convex combinations of
parameter-independent, deterministic distributions, I have also implicitly
assumed that no backwards causation occurs.%
}

Recall that a tight Bell inequality is an extremal, maximally restrictive
Bell inequality (see fn. \ref{fn:Tight-Bell-Inequalities}). Thus,
the set of all Bell inequalities can be characterized by the complete
set of tight Bell inequalities.

Geometrically, the set of Bell distributions (convex combinations
of parameter-independent, deterministic distributions) corresponds
to a particular polytope. The tight Bell inequalities correspond to
the facets of that polytope. Thus, geometric tools can aid in the
search for these inequalities.%
\footnote{\citet{Pitowsky1989,Pitowsky:1991} are the classic texts introducing
geometric methods to the study of Bell inequalities.%
} However, much of the geometric structure involved in characterizing
polytopes is irrelevant to the search for the Bell inequalities. In
the following, I will use a more abstract mathematical object, the
matroid, to isolate the structure important for the task at hand.

\subsection{Matroids; duality theorem\label{sub:Matroids}}

Matroids are mathematical objects that can be used to encode the combinatoric
aspects of geometric structures.%
\footnote{Matroids have a multitude of other uses, particularly in graph theory.
For background on matroids, see \citet{Oxley:2011fk} and \citet{Bjorner:1999uq}. %
} My use of matroid theory is relatively limited, so the details will
be kept to a minimum here. The reader can skip to section \ref{sub:TBIC}
without significant loss of comprehension, if desired.

Matroids come in two varieties, oriented and unoriented. The extra
structure provided by oriented matroids is needed here, so all matroids
described below should be understood as oriented.

A matroid can be characterized in a variety of ways. One way is to
begin with a base set, $E$, and then specify a set of \emph{bases}
that satisfy a particular set of axioms. Various sets can then be
defined, which carry names drawn from linear algebra and graph theory:
\emph{independent sets,} \emph{dependent sets, hyperplanes}, \emph{circuits,
etc. }

Two types of matroids will be useful here. \emph{Vector matroids}
encode the linear independence properties of a set of vectors. \emph{Affine
matroids} encode the affine dependencies of a set of points. Any set
of vectors defines a vector matroid, and any set of points, including
the vertices of a polytope, defines an affine matroid.%
\footnote{The bases of vector matroid are the maximal sets of linearly independent
vectors, and similarly with affine matroids.%
}

The task of finding the facets of a polytope is equivalent to that
of finding the positive hyperplanes of the corresponding affine matroid.

Matroids have duals, which are also matroids. The bases of a dual
matroid are the complements of the bases of the original.

Given a factorizable set $\Pi A_{B}$, and some subset $\Psi\subseteq\powerset B$,
one can construct a polytope by projecting the simplex in $\dsR^{\left|\Pi A_{B}\right|-1}$
defined by $\Pi A_{B}$ into the subspace defined by the collection
of multideviation vectors corresponding to the elements of $\Psi$.
One can then use this polytope to define an affine matroid, $\function{M_{A}}{\Psi}$.
After choosing an origin, one can define a vector matroid, $\function{M_{V}}{\Psi}$,
using the vectors pointing from the origin to the vertices.

I have been able to prove the following result concerning these matroids:
\begin{thm}[Duality of multideviation projections]
 \label{thm:-Matroid-duality}The affine matroid of a multideviation
projection $\Psi$ is the dual of the vector matroid of the complement
$\powerset B\backslash\Psi$ formed by taking the center of the polytope
as the origin:
\begin{equation}
\function{M_{A}^{*}}{\Psi}=\function{M_{V}}{\powerset B\backslash\left(\left\{ \emptyset\right\} \cup\Psi\right)}
\end{equation}

\end{thm}
The importance of this theorem for the task at hand cannot be understated.
Vector matroids are generally easier to work with than affine matroids.
Furthermore, in a well-known result in matroid theory, the hyperplanes
of a matroid are the complements of the circuits of the dual, and
circuits are generally easier to specify than hyperplanes.

So, the task of finding the tight Bell inequalities corresponds to
that of finding the positive circuits of the vector matroid defined
by the complement of $\Psi$, which is the set of observable multideviation
degrees of freedom.

\subsection{Necessary and sufficient conditions (TBIC)\label{sub:TBIC}}

Drawing on the duality theorem, the necessary and sufficient conditions
for a set $\Gamma\subseteq\SampleOmni$ to define a tight Bell inequality
over an event space $\SampleES$ are the following:
\begin{enumerate}
\item There is a function $f$ such that:
\begin{align}
\textrm{(a)} & \forall\itpgamma\in\SampleOmni\left[\function f{\itpgamma}=0\longleftrightarrow\itpgamma\notin\Gamma\right]\nonumber \\
\textrm{(b)} & \forall\rho\in\powerset{\cup M}\backslash\Psi\;\;\forall\itpgamma\in\Gamma\left[\multideviation{\rho}f{\itpgamma_{\rho}}{\none}=0\right]\label{eq:NS-conditions}\\
\textrm{(c)} & \forall\itpgamma\in\Gamma\left[\function f{\itpgamma}>0\right]\nonumber 
\end{align}
where $\Psi=\bigcup_{\itpp\in\SampleJMS}\bigcup_{\sigma\in\powerset V}\left\{ \itpp_{\sigma}\right\} $.
\item All functions meeting criterion 1 are scalar multiples of one another.
\end{enumerate}
I will refer to these as the \emph{tight Bell inequality conditions}
(TBIC).

The connection to the above matroid result is as follows: 1a requires
the function $f$ to be a faithful representation of the set $\Gamma$;
1b requires $\Gamma$ to be a dependent set; 2 requires $\Gamma$
to be minimally dependent (i.e., a circuit); and 1c requires $\Gamma$
to be all positive. Thus $\Gamma$ is a positive circuit of $\function{M_{V}}{\Psi^{c}}$.

Given such a set $\Gamma$ and function $f$, the corresponding inequality
is
\begin{equation}
\sum_{\itpgamma_{\cup M}}\function f{\itpgamma}\mcprobability{}{}{\itpgamma}\geq0
\end{equation}
which is equivalent to
\begin{equation}
\sum_{\sigma\in\Psi}\sum_{\itpgamma_{\sigma}}\multideviation{\sigma}f{\itpgamma_{\sigma}}{\none}\multideviation{\sigma}P{\itpgamma_{\sigma}}{\none}\geq0\label{eq:BI in terms of MDs}
\end{equation}
Using \eqref{eq:Mapping-between-OJ-and-MC}, we can write \eqref{eq:BI in terms of MDs}
in terms of the multiple-context distribution (i.e., the observed
statistics). What \eqref{eq:BI in terms of MDs} says is that when
the tight Bell inequalities are expressed in terms of multideviations
of the observed statistics, the multideviations of the linear dependence
function $f$ give the coefficients. 

Condition 1b of the TBIC allows a remarkable simplification, one that
will prove useful below. For any $i\in V$, $p_{i},q_{i}\in M_{i}$,
and $\itpgamma_{\unionM}$,

\begin{equation}
n_{p_{i}}n_{q_{i}}\function f{\itpgamma}-n_{q_{i}}\function{f^{\left\{ p_{i}\right\} }}{\itpgamma}-n_{p_{i}}\function{f^{\left\{ q_{i}\right\} }}{\itpgamma}+\function{f^{\left\{ p_{i},q_{i}\right\} }}{\itpgamma}=0\label{eq:LS-equations}
\end{equation}
where $\function{f^{\sigma}}{\itpgamma_{\sigma}}\equiv\sum_{\itpgamma_{V\backslash\sigma}}\function f{\itpgamma}$.
This form of the condition will be particularly useful in sections
\ref{sub:Simplification---lifting-up} and \ref{sec:Pioneer-sets}.
\eqref{eq:LS-equations-simplified} uses simplified notation to express
this in even simpler form.

\subsection{Simplification---lifting up\label{sub:Simplification---lifting-up}}

As \citet{Pironio2005} showed, any tight Bell inequality specifies
a similarly structured inequality for more observers, more observables,
and/or different types of observables. In the next section, I will
specify solutions for arbitrary numbers of observers each choosing
between two binary observables. Here I will specify the ``lifted
up'' solutions for greater numbers of observables or outcomes.

Consider an arbitrary event space, $\SampleESFull$, and a binary
event space with the same number of observers, $\dsTwo_{V}\equiv\left(V,M_{V}^{'},N_{\cup M_{V}}^{'}\right)$,
where $\left|M_{i}^{'}\right|=2$ for all $i\in V$ and $\left|N_{p_{i}}^{'}\right|=2$
for all $p_{i}\in\cup M_{V}$. Suppose $\Gamma\subseteq\Pi N_{\cup M_{V}^{'}}^{'}$
picks out a solution of the TBIC for $\dsTwo_{V}$, where $f$ is
the corresponding function. Now do the following:
\begin{enumerate}
\item Select some $\itpp_{V},\itpq_{V}\in\SampleJMS$ where $p_{i}\neq q_{i}$
for all $i\in V$. Let $PQ=\left(\itpp_{V}\cup\itpq_{V}\right)$. 
\item Relabel the elements of $M_{V}^{'}$ so that $M_{i}^{'}=\left\{ p_{i},q_{i}\right\} $
for all $i$. 
\item Let $L_{p_{i}}=\function{\scL}{N_{p_{i}}}$ be a boolean lattice over
$N_{p_{i}}$. Select some $\itpalpha_{PQ}\in\Pi L_{PQ}$.%
\footnote{The lattice intuple $\itpalpha_{PQ}$ represents a block of intuples
formed by taking the Cartesian product of the lattice elements. For
example, if $\alpha_{i}=\left\{ 1_{i},2_{i}\right\} $ and $\alpha_{j}=\left\{ 2_{j},3_{j}\right\} $,
then $\itpalpha_{\left\{ i,j\right\} }=\left\{ 1_{i}2_{j},1_{i}3_{j},2_{i}2_{j},2_{i}3_{j}\right\} $.%
}
\item Define a mapping function $\itp{\chi}:\SampleOmni\rightarrow\Pi N_{PQ}^{'}$
such that $\forall\mu_{i}\in PQ$,
\begin{equation}
\left(\function{\itp{\chi}}{\itpgamma}\right)_{\mu_{i}}=\begin{cases}
1_{\mu_{i}} & \gamma_{\mu_{i}}\in\alpha_{i}\\
2_{\mu_{i}} & \gamma_{\mu_{i}}\notin\alpha_{i}
\end{cases}
\end{equation}

\item Let $\Gamma^{*}=\left\{ \itpgamma\in\SampleOmni\;|\;\function{\itp{\chi}}{\itpgamma}\in\Gamma\right\} $.
\end{enumerate}
Let $f$ be a real-valued function over $\Pi N_{\cup M_{V}^{'}}^{'}$.
Then define $f^{*}$ such that
\begin{equation}
\function{f^{*}}{\itpgamma}=\function f{\itp{\chi}\left(\itpgamma\right)}\label{eq:f-star}
\end{equation}
where $\itpgamma\in\SampleOmni$.

If we substitute \eqref{eq:f-star} into \eqref{eq:LS-equations}
for the full event space, the resulting equations are identical to
\eqref{eq:LS-equations} for $\Gamma$ in the event space $\dsTwo_{V}$.
Thus, if $\Gamma$ satisfies the TBIC for $\dsTwo_{V}$, then $\Gamma^{*}$
must satisfy the TBIC for $\SampleESFull$.

This shows that Bell inequalities for cases where observers have many
choices of arbitrarily complicated observables can be generated from
Bell inequalities for cases where the same number of observers choose
between two binary observables. This does \emph{not} show that these
are the only Bell inequalities for the more complicated event spaces.%
\footnote{Some inequalities for more than 2 observables per observer and more
than 2 outcomes per observable that are not reducible in this way
are known (see \citealt{GargMermin:1982,CollinsGisin:2004}).%
}

The ``lifted up'' solutions have essentially the same structure
as the source solutions. All but a pair of observables for each observer
are ignored, and the outcome space for each observable is viewed as
binary (i.e., the outcome either is or is not in $\alpha_{p_{i}}$).
The resulting inequalities are thus best expressed in terms of Boolean
multideviations (see section \ref{sub:Boolean-multideviations}):
\begin{equation}
\sum_{\rho\in\Psi}\multideviation{\rho}f{\itp 1_{\sigma}}{\none}\function{W_{P}^{\rho,\unionM}}{\itpalpha_{\rho}}\geq0\label{eq:BI in terms of MDs-W}
\end{equation}
Note that the form is invariant under changes in the number or types
of observables.

\section{Pioneer sets---new tight Bell inequalities \label{sec:Pioneer-sets} }

In this section I will present a set of solutions to the TBIC for
arbitrary numbers of observers choosing between two binary observables.
As shown in section \ref{sub:Simplification---lifting-up}, these
will also generate solutions for arbitrarily complicated physical
scenarios. I refer to these solutions as ``pioneer sets'', for the
way they way they branch out through the outcome space. 

While the set of these solutions grows exponentially with the number
of observers, and while solutions with genuinely new structure exist
at each level, these are unfortunately but a small portion of the
set of all solutions to the TBIC.%
\footnote{For example, for 3 observers, there are 352 pioneer sets. However,
\citet{PitwoskySvozil2001} have shown the existence of 53856 facets.%
} However, these solutions are relatively easily characterized and
have a straightforward conceptual interpretation.  

The reader uninterested in the details of these solutions can skip
to section \ref{sub:Pioneer sets - Special-cases}, where the Bell
inequalities for a particularly simple subset of them are presented.

\subsection{Definition\label{sub:Pioneer-set-Definition}}

Given an event space $\dsTwo_{V}\equiv\left(V,M_{V},N_{\cup M_{V}}\right)$,
where $M_{i}=\left\{ p_{i},q_{i}\right\} $ and $N_{\mu_{i}}=\left\{ 1_{\mu_{i}},2_{\mu_{i}}\right\} $
for all $i\in V$ and $\mu_{i}\in M_{i}$, a pioneer set is characterized
by a pair $\left(Z,S_{Z}\right)$, where
\begin{enumerate}
\item $Z$ is a partition of $V$.
\item $S_{Z}$ is an indexed family of sets, where $S_{z}\subseteq\powerset z$,
for each $z\in Z$.
\item For each $z\in Z$, and each $i,j\in z$ where $i\ne j$, there is
a sequence of elements of $S_{z}$ such that $i$ is in the first
element, $j$ is in the last, and every pair of consecutive elements
has a non-empty intersection.%
\footnote{For example, if $z=\left\{ A,B,C\right\} $, then $S_{z}=\left(\left\{ A,B\right\} ,\left\{ B,C\right\} \right)$
would satisfy this requirement, while $S_{z}=\left(\left\{ A,B\right\} ,\left\{ C\right\} \right)$
would not.%
}
\end{enumerate}

\subsubsection{The odd-out transformation\label{sub:The-odd-out-transformation}}

The odd-out transformation applies to subsets of a powerset. 

Given a set $z$ and a set $S\subseteq\powerset z$, the odd-out transformation,
$S^{*}$, is
\begin{equation}
S^{*}\equiv\left\{ \sigma\in\powerset z\;:\;\left|S\cap\powerset{\sigma}\right|\textrm{is odd}\right\} 
\end{equation}
Note that $S^{**}=S$.

The odd-out transforms of the sets in the indexed family $S_{Z}$
will be represented $S_{Z}^{*}$.

\subsubsection{Relabeling outcomes; condition 1b of the TBIC}

Each element in the omni-joint outcome space $\SampleOmni$ can be
represented succinctly by two subsets $\sigma,\rho\subseteq V$:
\begin{equation}
\left(\sigma,\rho\right)\longleftrightarrow\itp 1_{\itpp_{V\backslash\left(\sigma\cminus\rho\right)}\itpq_{V\backslash\rho}}\itp 2_{\itpp_{\left(\sigma\cminus\rho\right)}\itpq_{\rho}}
\end{equation}
With this convention, \eqref{eq:LS-equations}, which is equivalent
to 1b of the TBIC, takes a particularly simple form:
\begin{equation}
\function f{\sigma,\rho}+\function f{\sigma,\rho\cminus\left\{ i\right\} }=\function f{\sigma\cminus\left\{ i\right\} ,\rho}+\function f{\sigma\cminus\left\{ i\right\} ,\rho\cminus\left\{ i\right\} }\label{eq:LS-equations-simplified}
\end{equation}
for any $i\in V$ and $\sigma,\rho\seq V$.

\subsubsection{The pioneer set}

Let $\pioneerset$ be the pioneer set characterized by $\left(Z,S_{Z}\right)$.
Then 
\begin{equation}
\left(\sigma,\rho\right)\in\pioneerset\longleftrightarrow\forall z\in Z\left[\left(\sigma\cap z,\rho\cap z\right)\in\pioneerset_{z}\right]
\end{equation}
where
\begin{equation}
\left(\mu,\nu\right)\in\pioneerset_{z}\longleftrightarrow\left(\mu\in S_{z}^{*}\leftrightarrow\left|\nu\right|\textrm{is odd}\right)\label{eq:Pioneer-set-In-out}
\end{equation}

\subsection{The corresponding inequalities}

Proof that pioneer sets define tight Bell inequalities is given in
Appendix \ref{sec:Proof-that-pioneer-satisfy-the-TBIC}.

\subsubsection{General case}

To specify the inequalities, which are given by \eqref{eq:BI in terms of MDs}
and \eqref{eq:BI in terms of MDs-W}, we need only give the multideviation
of $f$ for each element of $\Psi$:
\begin{equation}
\multideviation{\itpp_{\sigma\backslash\rho}\itpq_{\rho}}f{\itp 1_{\itpp_{\sigma\backslash\rho}\itpq_{\rho}}}{PQ}=\frac{1}{2^{\text{\ensuremath{\left|Z\right|}}}}\prod_{z\in Z}\left(\delta_{z\cap\sigma=\emptyset}+\frac{1}{2^{\left|z\right|}}\delta_{z\backslash\sigma=\emptyset}\sum_{\mu\in\powerset z}\left(-1\right)^{\left|\mu\backslash\rho\right|}\left(-1\right)^{\delta_{\mu\in S_{z}^{*}}}\right)
\end{equation}
where $\sigma\seq V$ and $\rho\seq\sigma$. 

When $\left|Z\right|>1$, the inequality is a straightforward composition
of lower-level inequalities. In other words, inequalities for two
sets of observers $V$ and $V^{'}$ always define an inequality for
$V\un V^{'}$.%
\footnote{This is likely true in general, not just for pioneer sets. Note that
this is a stronger claim than that found in \citet{Pironio2005},
which only applies to cases where either $\left|V\right|=1$ or $\left|V^{'}\right|=1$.%
}

\subsubsection{Special cases; the even/odd inequalities\label{sub:Pioneer sets - Special-cases}}

Genuinely new structure for each $V$ is found when $Z=\left\{ V\right\} $.
In this case, the inequality coefficients take a simpler form:
\begin{equation}
\multideviation{\itpp_{\sigma\backslash\rho}\itpq_{\rho}}f{\itp 1_{\itpp_{\sigma\backslash\rho}\itpq_{\rho}}}{PQ}=\frac{1}{2}\left(\delta_{\sigma=\emptyset}+\frac{1}{2^{\left|V\right|}}\delta_{\sigma=V}\sum_{\mu\in\powerset V}\left(-1\right)^{\left|\mu\backslash\rho\right|}\left(-1\right)^{\delta_{\mu\in S^{*}}}\right)\label{eq:BIs-top-level}
\end{equation}
Note that, aside from the constant $\multideviation{\emptyset}{\none}{\none}{\none}$,
only the top-level multideviations ($\sigma=V$) will appear in the
inequality. As I will show in section \ref{sec:Philosophical-importance-of-EO-inqs},
this has important philosophical consequences.

For reasons that will be made clear in section \ref{sec:Intepretation-of-the-even/odd-inequalities},
I refer to these as \emph{even/odd inequalities}.

From this, a series of particularly simple inequalities can be derived.
For any $\varphi\seq V$ and $m\in\left\{ 0,1\right\} $,

\begin{equation}
\frac{1}{2}+\left(-1\right)^{m}\left(2^{\left|V\right|-1}\multideviation{\itpp_{V\backslash\varphi}\itpq_{\varphi}}P{\none}{\none}-\sum_{\rho\in\powerset V}\multideviation{\itpp_{V\backslash\rho}\itpq_{\rho}}P{\none}{\none}\right)\geq0\label{eq:Simplest BIs}
\end{equation}
where $\multideviation{\mu}P{\none}{\none}=\multideviation{\mu}P{\itp 1_{\mu}}{\none}$
by convention. For general event spaces, substitute $\multideviation{\mu}P{\itp 1_{\mu}}{\none}\rightarrow\lmultideviation{\mu}P{\itpalpha_{\mu}}{\none}$,
as in \eqref{eq:BI in terms of MDs-W}.

When $\left|V\right|=2$, \eqref{eq:Simplest BIs} reduces to the
CHSH inequality (see section \ref{sub:CHSH}).

If there were no restrictions on the base states, then the $\multideviation{\itpp_{V\backslash\rho}\itpq_{\rho}}P{\none}{\none}$
could range between $\pm\frac{1}{2^{\left|V\right|}}$ independently
of one another. The left-hand side of \eqref{eq:Simplest BIs} would
then be able to reach $-1$, the maximal violation allowed by the
probability calculus. 

In section \ref{sec:Violations in QM}, I will show that quantum mechanics
predicts a violation of each of these inequalities.

\subsection{Counts}

Some data on the pioneer sets is given in table \ref{tab:data on pioneer sets}.
The total number grows roughly as $2^{2^{\left|V\right|}}$, and those
that show genuinely new structure for a given number of observers
(i.e., those where $Z=\left|V\right|$) quickly come to dominate.
\begin{table}[h]
\begin{centering}
\begin{tabular}{|c|c|c|c|}
\hline 
$\left|V\right|$ & \# of pioneer sets & \# with $Z=\left\{ V\right\} $ & $2^{2^{\left|V\right|}}$\tabularnewline
\hline 
\hline 
2 & 24 & 8 & 16\tabularnewline
\hline 
3 & 352 & 192 & 256\tabularnewline
\hline 
4 & 67,968 & 63680 & 65536\tabularnewline
\hline 
5 & $\sim4.296\times10^{9}$ & $\sim4.294\times10^{9}$ & $\sim4.295\times10^{9}$\tabularnewline
\hline 
6 & $\sim1.845\times10^{19}$ & $\sim1.845\times10^{19}$ & $\sim1.845\times10^{19}$\tabularnewline
\hline 
\end{tabular}
\par\end{centering}

\caption{\label{tab:data on pioneer sets}Counts of pioneer sets by number
of observers.}
\end{table}

\section{Conceptual intepretation of the even/odd inequalities\label{sec:Intepretation-of-the-even/odd-inequalities}}

The interpretation of the binary multideviations given in section
\ref{sub:Even-and-odd} allows a straightforward interpretation of
the even/odd inequalities, along the lines described in section \ref{sub:CHSH}
for the CHSH inequality. The new inequalities represent limits on
how incompatible the even or odd statistics for the joint measurement
contexts can be with one another.

Let a distribution be \emph{odd-definite} if it will, with certainty,
produce a joint outcome with an odd number of 2's, and \emph{even-definite}
if it will, with certainty, produce a joint outcome with an even number
of 2's.%
\footnote{That is, the distribution $P_{\itp{\mu}}$ is odd- or even-definite
if $\function{Pr_{\itp{\mu}}}{\textrm{even \# of \ensuremath{V}\ outcomes are }2}$
is $0$ or $1$, respectively.%
} Then the even/odd inequalities express logical connections between
the odd/even-definiteness of the distributions for different measurement
contexts.

To see how these connections arise, consider a joint measurement of
three binary observables, labeled 1, 2, and 3. If the distribution
is even-definite in observables 1 and 2 and even-definite in observables
1 and 3, then it is necessarily even-definite in observables 2 and
3. For example, suppose the results of three coins being flipped are
even-definite in coins 1 and 2, meaning that the joint outcome must
include either $H_{1}H_{2}$ or $T_{1}T_{2}$. Suppose further that
the results are even-definite in 1 and 3, so that the joint outcome
must include $H_{1}H_{3}$ or $T_{1}T_{3}$. The total joint outcome
must thus be either $H_{1}H_{2}H_{3}$ or $T_{1}T_{2}T_{3}$, and
the results must thus be even-definite in 2 and 3. In other words,
the requirements that the results be even-definite in 12 and also
in 13 imply that the results are also even-definite in 23.

We could have phrased this constraint just as easily in terms of odd-definiteness:
it cannot be the case that an odd number of pairs of 1, 2, and 3 are
odd-definite. There will be four such constraints---one for each subset
of $\left\{ 12,13,23\right\} $ with an odd number of elements.

These limits can be represented through a simple graphical method,
depicted in fig. \ref{fig:Even-odd compatibility}. They represent
possible deterministic arrangements, where the outcome of each observable
is represented as $\pm1$. The multideviations for pairs of observables
are determined by multiplying the values for the corresponding observables.
Certain arrangements among the pairwise multideviations cannot be
formed in this way, namely, those where an odd number have the value
$-1$. These represent logical restrictions on the distribution. 
\begin{figure}
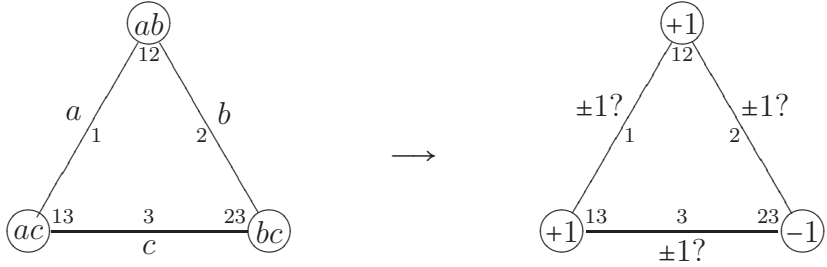

\[
\xy 
(0,-10)="O";
(8,0)="X2";
(0,13.8)="Y2";
{"Y2"+"Y2"}="Y";
{"X2"+"X2"}="X";
{"O"+"Y"}*+{ab}*\frm{o}="V12";
{"O"-"X"}*+{ac}*\frm{o}="V13";
{"O"+"X"}*+{bc}*\frm{o}="V23";
"V12"!(0,-2)*+!U{\scriptstyle{12}};
"V13"!(2,2)*+!L{\scriptstyle{13}};
"V23"!(-2,2)*+!R{\scriptstyle{23}};
"V12";"V13"**\dir{-}?*!(2,-2){a}?*!(-1,1){\scriptstyle{1}};
"V12";"V23"**\dir{-}?*!(-2,-2){b}?*!(1,1){\scriptstyle{2}};
"V13";"V23"**\dir{-}?*!(0,2){c}?*!(0,-2){\scriptstyle{3}};
\endxy
\;\;\;\;\;\;\;\;\;\;\;\;
\longrightarrow
\;\;\;\;\;\;\;\;\;\;\;\;
\xy 
(0,-10)="O";
(8,0)="X2";
(0,13.8)="Y2";
{"Y2"+"Y2"}="Y";
{"X2"+"X2"}="X";
{"O"+"Y"}*+{+1}*\frm{o}="V12";
{"O"-"X"}*+{+1}*\frm{o}="V13";
{"O"+"X"}*+{-1}*\frm{o}="V23";
"V12"!(0,-2)*+!U{\scriptstyle{12}};
"V13"!(2,2)*+!L{\scriptstyle{13}};
"V23"!(-2,2)*+!R{\scriptstyle{23}};
"V12";"V13"**\dir{-}?*!(3,-3){\pm1?}?*!(-1,1){\scriptstyle{1}};
"V12";"V23"**\dir{-}?*!(-3,-3){\pm1?}?*!(1,1){\scriptstyle{2}};
"V13";"V23"**\dir{-}?*!(0,2.5){\pm1?}?*!(0,-2){\scriptstyle{3}};
\endxy 
\]

\caption{\label{fig:Even-odd compatibility}Even/odd compatibility for 3 observables.
A vertex arrangement of $\pm1$ can be generated by assigning $\pm1$
to each side and then multiplying the adjacent values, as in the diagram
on the left. Certain vertex arrangements of $\pm1$ cannot be constructed
in this way. The diagram on the right is an example of such an arrangement.}
\end{figure}

The even/odd inequalities are logical limitations in precisely the
same way. For the 2-observer case, the corresponding graph is a square,
and the impossible arrangements are also those for which an odd number
of vertices have a $-1$. There are 8 such arrangements, corresponding
to the 8 CHSH inequalities.

For more observers, the same restriction holds---the distribution
cannot be odd-definite in an odd number of contexts. However, for
3 observers or more, there can also be more complicated forms of even/odd
incompatibility. One such arrangement for 3 observers is depicted
in fig. \ref{fig:Even-odd-limitation-for-three-observers}. 

\begin{figure}
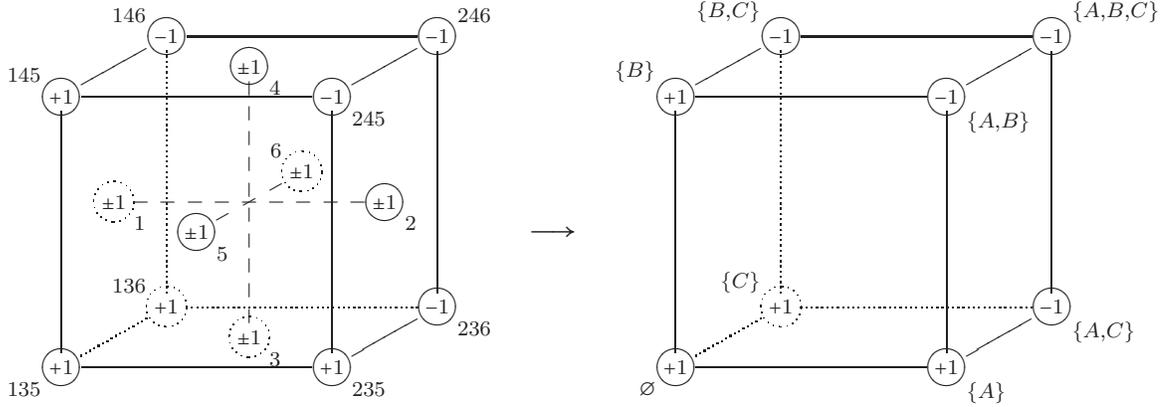

\begin{centering}
\begin{center}
\[
\xy 
(0,0)="O";
(9,0)="X2";
(0,9)="Y2";
"X2"+"X2"="X";
"Y2"+"Y2"="Y";
(14,0)="dX";
(0,8)="dY";
{"O"+"X"+"Y"}*+{\scriptstyle{-1}}*\frm{o}="bba";
{"O"+"X"-"Y"}*+{\scriptstyle{+1}}*\frm{o}="baa";
{"O"-"X"+"Y"}*+{\scriptstyle{+1}}*\frm{o}="aba";
{"O"-"X"-"Y"}*+{\scriptstyle{+1}}*\frm{o}="aaa";
{"O"+"X"+"Y"+"dX"+"dY"}*+{\scriptstyle{-1}}*\frm{o}="bbb";
{"O"+"X"-"Y"+"dX"+"dY"}*+{\scriptstyle{-1}}*\frm{o}="bab";
{"O"-"X"+"Y"+"dX"+"dY"}*+{\scriptstyle{-1}}*\frm{o}="abb";
{"O"-"X"-"Y"+"dX"+"dY"}*+{\scriptstyle{+1}}*\frm{.o}="aab";
"aaa"!DL*!UR{\scriptstyle{135}};
"aba"!UL*!DR{\scriptstyle{145}};
"baa"!DR*!UL{\scriptstyle{235}};
"bba"!DR*!UL{\scriptstyle{245}};
"aab"!UL*!DR{\scriptstyle{136}};
"abb"!UL*!DR{\scriptstyle{146}};
"bab"!DR*!UL{\scriptstyle{236}};
"bbb"!UR*!DL{\scriptstyle{246}};
"aaa";"abb"**@{}?*{\scriptstyle{\pm1}}*+\frm{.o}="acc";
"baa";"bbb"**@{}?*+{\scriptstyle{\pm1}}*\frm{o}="bcc";
"aaa";"bab"**@{}?*{\scriptstyle{\pm1}}*+\frm{.o}="cac";
"aba";"bbb"**@{}?*+{\scriptstyle{\pm1}}*\frm{o}="cbc";
"aaa";"bba"**@{}?*{\scriptstyle{\pm1}}*+\frm{o}="cca";
"aab";"bbb"**@{}?*+{\scriptstyle{\pm1}}*\frm{.o}="ccb";
"acc"!DR*!UL{\scriptstyle{1}};
"bcc"!DR*!UL{\scriptstyle{2}};
"cac"!DR*!UL{\scriptstyle{3}};
"cbc"!DR*!UL{\scriptstyle{4}};
"cca"!DR*!UL{\scriptstyle{5}};
"ccb"!UL*!DR{\scriptstyle{6}};
"aaa";"aba"**@{-};
"aab";"abb"**@{.};
"baa";"bba"**@{-};
"bab";"bbb"**@{-};
"aaa";"baa"**@{-};
"aba";"bba"**@{-};
"aab";"bab"**@{.};
"abb";"bbb"**@{-};
"aaa";"aab"**@{.};
"baa";"bab"**@{-};
"aba";"abb"**@{-};
"bba";"bbb"**@{-};
"acc";"bcc"**@{--};
"cac";"cbc"**@{--};
"cca";"ccb"**@{--};
\endxy
\;\;\;\;
\longrightarrow
\;\;\;\;
\xy 
(0,0)="O";
(9,0)="X2";
(0,9)="Y2";
"X2"+"X2"="X";
"Y2"+"Y2"="Y";
(14,0)="dX";
(0,8)="dY";
{"O"+"X"+"Y"}*+{\scriptstyle{-1}}*\frm{o}="bba";
{"O"+"X"-"Y"}*+{\scriptstyle{+1}}*\frm{o}="baa";
{"O"-"X"+"Y"}*+{\scriptstyle{+1}}*\frm{o}="aba";
{"O"-"X"-"Y"}*+{\scriptstyle{+1}}*\frm{o}="aaa";
{"O"+"X"+"Y"+"dX"+"dY"}*+{\scriptstyle{-1}}*\frm{o}="bbb";
{"O"+"X"-"Y"+"dX"+"dY"}*+{\scriptstyle{-1}}*\frm{o}="bab";
{"O"-"X"+"Y"+"dX"+"dY"}*+{\scriptstyle{-1}}*\frm{o}="abb";
{"O"-"X"-"Y"+"dX"+"dY"}*+{\scriptstyle{+1}}*\frm{.o}="aab";
"aaa"!DL*!UR{\scriptstyle{\emptyset}};
"aba"!UL*!DR{\scriptstyle{\{B\}}};
"baa"!DR*!UL{\scriptstyle{\{A\}}};
"bba"!DR*!UL{\scriptstyle{\{A,B\}}};
"aab"!UL*!DR{\scriptstyle{\{C\}}};
"abb"!UL*!DR{\scriptstyle{\{B,C\}}};
"bab"!DR*!UL{\scriptstyle{\{A,C\}}};
"bbb"!UR*!DL{\scriptstyle{\{A,B,C\}}};
"aaa";"aba"**@{-};
"aab";"abb"**@{.};
"baa";"bba"**@{-};
"bab";"bbb"**@{-};
"aaa";"baa"**@{-};
"aba";"bba"**@{-};
"aab";"bab"**@{.};
"abb";"bbb"**@{-};
"aaa";"aab"**@{.};
"baa";"bab"**@{-};
"aba";"abb"**@{-};
"bba";"bbb"**@{-};
\endxy
\]
\par\end{center}
\par\end{centering}

\begin{centering}

\par\end{centering}

\caption{\label{fig:Even-odd-limitation-for-three-observers}Example of an
even-odd limitation for three observers, where $\itpp=\left\{ 1,3,5\right\} $
and $\itpq=\left\{ 2,4,6\right\} $. The vertices represent joint
measurement contexts, and the distribution is even-definite if $+1$
is assigned and odd-definite if $-1$ is assigned. There is no assignment
of $\pm1$ to the facets of the cube such that the product of adjacent
facets at each vertex produces the above arrangement. The joint measurement
contexts can be labeled by subsets of $\left\{ A,B,C\right\} $. The
set of odd-definite contexts is then $\left\{ \left\{ A,B\right\} ,\left\{ A,C\right\} ,\left\{ B,C\right\} ,\left\{ A,B,C\right\} \right\} $,
and the odd-out transform is $\left\{ \left\{ A,B\right\} ,\left\{ A,C\right\} ,\left\{ B,C\right\} \right\} $,
which defines the corresponding pioneer set. The corresponding Bell
inequality is given by \eqref{eq:BIs-top-level}: $\frac{1}{4}-\simplemult{\left\{ 1,3,5\right\} }+\simplemult{\left\{ 1,4,6\right\} }+\simplemult{\left\{ 2,3,6\right\} }+\simplemult{\left\{ 2,4,5\right\} }\geq0$.}

\end{figure}

The graphical method generalizes straightforwardly to $n$ observers,
where the relevant graph is an $n$-dimensional hypercube. Each even/odd
inequality corresponds to an assignment of $\pm1$ to the vertices
that cannot be generated by assigning $\pm1$ to each facet and placing
the product of adjacent facets at each vertex. The profile $S$ which
defines the pioneer set for the inequality is just the odd-out transform
of the set of odd-definite contexts, when each is labeled as a subset
of the set of observers (as in fig. \ref{fig:Even-odd-limitation-for-three-observers}). 

Indeed, every inequality defined by the pioneer sets represents a
kind of even/odd incompatibility, although only the even/odd inequalities
can be represented so easily.

\section{Philosophical importance of the even/odd inequalities\label{sec:Philosophical-importance-of-EO-inqs}}

The even/odd inequalities are interesting not merely because they
have a convenient conceptual interpretation, but also because they
provide the opportunity for a stronger version of Bell's theorem.
Because these inequalities concern degrees of freedom that are independent
of those involved in parameter independence, they should be derivable
from conditions that do not include or imply parameter independence.
Violations of the even/odd inequalities could thus provide more specific
demands on which classical concepts must be given up. It is not my
intention here to provide such a derivation, merely to show that this
is possible.

As above, I will focus on the binary case, $\dsTwo_{V}$, with the
proviso that all results hold for the general case under the substitution
$\multideviation{\mu}P{\itp 1_{\mu}}{\none}\rightarrow\lmultideviation{\mu}P{\itpalpha_{\mu}}{\none}$,
as in \eqref{eq:BI in terms of MDs-W}. In the binary case, parameter
independence is
\begin{equation}
\forall\rho\supseteq\sigma\left[\multideviation{\sigma}{\itpp}{\none}{\none}=\multideviation{\sigma}{\itpp_{\rho}\itpq_{V\backslash\rho}}{\none}{\none}\right]
\end{equation}
and outcome independence is
\begin{equation}
\multideviation{\sigma}{\itpp}{\none}{\none}=\frac{1}{2^{\left|V\right|}}\prod_{i\in\sigma}\left(2^{\left|V\right|}\multideviation{\left\{ i\right\} }{\itpp}{\none}{\none}\right)
\end{equation}
where $\multideviation{\sigma}{\itp{\mu}}{\none}{\none}\equiv\multideviation{\itp{\mu}_{\sigma}}{P_{\itp{\mu}}}{\itp 1_{\itp{\mu}_{\sigma}}}{\none}$.

The important thing to notice is that parameter independence does
not affect any of the multideviations of the form $\multideviation V{\itp{\mu}}{\none}{\none}$
(i.e., when $\sigma=V$). On the other hand, outcome independence
does affect these degrees of freedom. Indeed, it is through the combination
of the two that the various $\multideviation V{\itp{\mu}}{\none}{\none}$
are related to one another; outcome independence relates $\multideviation V{\itp{\mu}}{\none}{\none}$
to the $\multideviation{\left\{ i\right\} }{\itp{\mu}}{\none}{\none}$
within a given measurement context, and parameter independence relates
the $\multideviation{\left\{ i\right\} }{\itp{\mu}}{\none}{\none}$
between measurement contexts. Fig. \ref{fig:Degrees-of-freedom-for-mcd-in-2x2x2}
depicts these relationships for the 2 observer case.

In the notation used in this section, the even/odd inequalities are
\begin{equation}
1+\sum_{\rho\in\powerset V}\multideviation V{\itpp_{V\backslash\rho}\itpq_{\rho}}{\none}{\none}\left(\sum_{\mu\in\powerset V}\left(-1\right)^{\left|\mu\backslash\rho\right|}\left(-1\right)^{\delta_{\mu\in S^{*}}}\right)\ge0
\end{equation}
where $S^{*}$ is the odd-out transform of the profile set $S\seq\powerset V$
(see \ref{sub:Pioneer-set-Definition}). It is clear that the inequalities
concern only multideviations of the form $\multideviation V{\itp{\mu}}{\none}{\none}$.
Thus, violations of the inequalities have nothing directly to do with
parameter independence. They test a type of locality (or other classical
concept) that does not involve effects of the choice of measurement
context.

The even/odd inequalities provide an opportunity; one could, in theory,
derive them from a condition placed solely on multideviations of the
form $\multideviation V{\itp{\mu}}{\none}{\none}$. Since these inequalities
are violated by quantum mechanics (to be shown in the next section),
one could then conclude that this condition cannot be satisfied by
any theory aiming to reproduce quantum statistics. Since this condition
would be manifestly independent of parameter independence, the result
would be strictly stronger than the existing Bell's theorem beginning
with parameter and outcome independence.

The trick, of course, is to find such a condition with a natural physical
interpretation. Because multideviations are new quantities without
well-established interpretations, there is not an obvious candidate
at this time.

\section{Quantum mechanics\label{sec:Violations in QM}}

I will now show that the inequalities presented in section \ref{sub:Pioneer sets - Special-cases}
are violated by quantum mechanics. Contrary to what one might expect,
the size of the violations increases with the number of observers
and converges toward the theoretical maximum.

\subsection{Experimental setup; initial state}

A set of observers, $V=\left\{ A,B,C,...\right\} $, each performs
one of two possible spin measurements, $M_{i}=\left\{ \theta_{i,0},\theta_{i,1}\right\} $,
in the $xz$-plane on spin-$\frac{1}{2}$ particles emitted from a
central source. There are thus two possible outcomes for each measurement,
and the event space has the structure of $\dsTwo_{V}$ (see section
\ref{sub:Pioneer-set-Definition}). 

The observables $\theta_{i,n}$ correspond to angles in the $xz$-plane,
where 0 represents the positive direction of the $z$-axis and $\frac{\pi}{2}$
represents the positive direction of the $x$-axis. The elements of
the outcome set, $N_{\theta_{i,n}}=\left\{ 1_{\theta_{i,n}},2_{\theta_{i,n}}\right\} $,
represent spin-up and spin-down, respectively, for the observable
$ $$\theta_{i,n}$. 

The overall Hilbert space for the experiment is the tensor product
of the 2-dimensional Hilbert spaces for each observer. States and
operators will be expressed in the positive $z$-basis of each subspace,
$\left\{ \ket{1_{i}},\ket{2_{i}}\right\} $ for observer $i$.

The initial state will be prepared in an ``even-correlation'' state:
\begin{equation}
\ket{\psi}=\frac{1}{\sqrt{2^{\left|V\right|-1}}}\underset{\left|\sigma\right|\text{ even}}{\sum_{\sigma\in\powerset V}}\left(-1\right)^{\frac{\left|\sigma\right|}{2}}\prod_{i\in V}\left(\delta_{i\notin\sigma}\ket{1_{i}}+\delta_{i\in\sigma}\ket{2_{i}}\right)
\end{equation}
If all observers measure along the positive $z$-axis, then the joint
outcome will always have an even number of spin-downs. As shown in
the next section, this state has a remarkably simple multideviation
profile, regardless of which spin-orientations are measured.

It is worth keeping in mind the difference between this state and
the generalized GHZ state, which is often used to represent multi-party
entanglement:
\begin{equation}
\ket{\psi_{\pm}}=\frac{1}{\sqrt{2}}\left(\left(\prod_{i\in V}\ket{1_{i}}\right)\pm\left(\prod_{i\in V}\ket{2_{i}}\right)\right)
\end{equation}
The GHZ states are perfectly correlated in a pairwise way. If any
two observers measure along the $z$-axis, then they will get the
same result. The multideviation profile for the GHZ state is significantly
more complicated than for the even-correlation state (only odd-order
multideviations vanish), especially for arbitrary spin-orientations.
Whether the GHZ states violate any of the top-level inequalities specified
by \eqref{eq:BIs-top-level} is unclear (for 3 or more observers).

\subsection{Measurement results}

A joint measurement context, $\itp{\mu}\in\Pi M_{V}$, can be specified
by the intuple $\itp m\in\prod_{i\in V}\left\{ 0,1\right\} $, where
$\mu_{i}=\theta_{i,m_{i}}$. We will thus consider the measurement
context to be a function of $\itp m$: $\function{\itp{\mu}}{\itp m}$.

The joint probabilities are 
\begin{equation}
\mcprobability{}{\function{\itp{\mu}}{\itp m}}{\itpx}=\frac{1}{2^{\left|V\right|}}\left(1+\left(-1\right)^{\left|\itpx\cap\itp 2\right|}\function{cos}{\sum_{i\in V}\theta_{i,m_{i}}}\right)
\end{equation}
The multideviations are
\begin{align}
\multideviation{\emptyset}{P_{\itp{\mu}}}{\itpx}{\itp{\mu}} & =\frac{1}{2^{\left|V\right|}}\\
\multideviation{\itp{\mu}}{P_{\itp{\mu}}}{\itpx}{\itp{\mu}} & =\frac{\left(-1\right)^{\left|\itpx\cap\itp 2\right|}}{2^{\left|V\right|}}\function{cos}{\sum_{i\in V}\theta_{i,m_{i}}}
\end{align}
where $\itp{\mu}=\function{\itp{\mu}}{\itp m}$. All other multideviations
vanish. 

Because the multideviations for different $\itpx$ differ only by
(at most) a factor of -1, we will focus only on $\multideviation{\sigma}{P_{\itp{\mu}}}{\itp 1}{\itp{\mu}}$.

\subsection{Violation of the simplest Even/Odd inequalities}

The simplest even/odd inequalities, given by \eqref{eq:Simplest BIs},
are indexed by $\varphi\seq V$ and $m\in\left\{ 0,1\right\} $. The
state $\ket{\psi}$ violates each one for certain choices of measurement
settings. 

Let 
\begin{align}
a_{i} & \equiv\theta_{i,\delta_{i\in\varphi}}\\
d_{i} & \equiv\frac{1}{2}\left(\theta_{i,\delta_{i\notin\varphi}}-a_{i}\right)
\end{align}
Then, after some algebraic manipulation, \eqref{eq:Simplest BIs}
becomes
\begin{equation}
\frac{1}{2}+\left(-1\right)^{\left|m\right|}\frac{1}{2}\left(\function{cos}{\sum_{i\in V}a_{i}}-2\,\function{cos}{\sum_{i\in V}\left(a_{i}+d_{i}\right)}\left(\prod_{i}^{V}\function{cos}{d_{i}}\right)\right)\geq0
\end{equation}
Now, let 
\begin{align}
a & \equiv\left(\sum_{i\in V}a_{i}\right)+\pi\delta_{m=0}\\
d_{i} & =\frac{\pi}{2\left|V\right|}\label{eq:substitution for d}
\end{align}
and the inequality becomes
\begin{equation}
\frac{1}{2}-\left(\frac{1}{2}\cos a+\sin a\:\cos^{\left|V\right|}\!\!\left(\frac{\pi}{2\left|V\right|}\right)\right)\geq0\label{eq:QM-violations}
\end{equation}
When $a=0$, the left-hand side is $0$, and the inequality is thus
satsified through equality. However, the derivative with respect to
$a$ is
\begin{equation}
-\left(-\frac{1}{2}\sin a+\cos a\:\cos^{\left|V\right|}\!\!\left(\frac{\pi}{2\left|V\right|}\right)\right)
\end{equation}
At $a=0$, this is $-\cos^{\left|V\right|}\!\!\left(\frac{\pi}{2\left|V\right|}\right)$,
which is manifestly negative. Thus, for values of $a$ slightly larger
than $0$, the inequality is violated.

\subsection{Maximal violations}

If, instead of \eqref{eq:substitution for d}, we assume $d_{i}=\frac{d}{\left|V\right|}$,
then we get the inequality
\begin{equation}
\frac{1}{2}-\left(\frac{1}{2}\cos a-\cos\left(a+d\right)\:\cos^{\left|V\right|}\!\!\left(\frac{d}{\left|V\right|}\right)\right)\geq0
\end{equation}
This expression is minimized over variations in $a$ when $a=\pi-d\left(\frac{\left|V\right|+1}{\left|V\right|}\right)$.
Then it becomes
\begin{equation}
\frac{1}{2}-\left(-\frac{1}{2}\cos\left(d+\frac{d}{\left|V\right|}\right)+\cos^{\left|V\right|+1}\!\!\left(\frac{d}{\left|V\right|}\right)\right)\label{eq:QM-Violations-with-d-open}
\end{equation}
Finding the minimum is difficult analytically, but a numerical search
is straightforward (see table \ref{tab:Maximal-QM-violations}). 

\begin{table}[h]
\begin{centering}
\begin{tabular}{|c|c|c|}
\hline 
$\left|V\right|$ & $\frac{d}{\pi}$ & Value of \eqref{eq:QM-Violations-with-d-open}\tabularnewline
\hline 
\hline 
$2$ & $0.5$ & $-0.207$\tabularnewline
\hline 
$3$ & $0.588$ & $-0.333$\tabularnewline
\hline 
$4$ & $0.689$ & $-0.421$\tabularnewline
\hline 
$5$ & $0.802$ & $-0.487$\tabularnewline
\hline 
$10$ & $0.972$ & $-0.669$\tabularnewline
\hline 
$100$ & $0.997$ & $-0.953$\tabularnewline
\hline 
$1000$ & $0.999$ & $-0.999$\tabularnewline
\hline 
\end{tabular}
\par\end{centering}

\caption{\label{tab:Maximal-QM-violations}Maximal violations of \eqref{eq:QM-violations}.}

\end{table}
Most important, the maximal violation of the inequality increases
with $\left|V\right|$. As $\left|V\right|\rightarrow\infty$, the
expression converges to 
\begin{equation}
\frac{1}{2}\left(\cos d-1\right)
\end{equation}
This is obviously minimized when $d=\pi$, where it is equal to $-1$.
As noted in section \ref{sub:Pioneer sets - Special-cases}, that
is the maximal violation allowed by the probability calculus (i.e.,
where there are no restrictions on the underlying distributions).

\section{Conclusion}

The introduction of multideviations exposes some of the underlying
structure of the distributions described by Bell's theorem. In particular,
those distributions can be generated from joint distributions over
all observables by ignoring specific multideviation degrees of freedom,
namely, those involving pairs of mutually exclusive observables. Thus,
further study of multideviations should help illuminate the philosophical
importance of Bell's theorem.

The new method for finding tight Bell inequalities presented above
does reduce the computational complexity of the problem somewhat,
but not enough to keep brute force calculations from being intractable
for relatively small numbers of observers. Still, the new organization
of the problem may prompt further improvements.

The presentation of new tight Bell inequalities for arbitrary numbers
of observers, particularly the even/odd inequalities, which have relatively
simple form and admit convenient conceptual interpretation, allows
for the confirmation that quantum mechanics violates the assumptions
of Bell's theorem (however they are formulated) for any number of
systems. Furthermore, the size of this violation increases with the
number of systems and converges to the theoretical maximum.

This last fact is somewhat surprising, for two different reasons.
First, quantum effects tend to be dampened in general as the number
of systems is increased. Yet, if we take the violation of a Bell inequality
to indicate something peculiarly non-classical about an experiment,
then the effect is \emph{more} pronounced as the number of systems
increases. 

Second, the fact that quantum mechanics does not permit a maximal
violation of the CHSH inequality has sparked a significant amount
of interest.%
\footnote{For some recent work, see \citet{FilippSvozil:2005,Janotta:2011by}.%
} The hope has been that some physical principle will explain the limitation
and perhaps provide some non-empirical justification for the Schr�dinger
equation. The above result, which shows that violation of Bell inequalities
converges toward the theoretical maximum as the number of systems
is increased, suggests that this limitation is a peculiarity of lower-dimensional
systems and perhaps of less fundamental importance than it may seem.

Finally, the even/odd inequalities concern degrees of freedom that
are unaffected by parameter independence, raising the possibility
of a new Bell's theorem that omits this condition altogether. Such
a theorem would allow a stronger philosophical result, namely, a more
precise articulation of the classical concepts that cannot be part
of any empirically adequate future physics.

\bibliographystyle{apalike}

\bibliography{BellsTheoremRefs}

\begin{thebibliography}{}

\bibitem[Bancal et~al., 2011]{BancalGisin:2011}
Bancal, J.-D., Brunner, N., Gisin, N., and Liang, Y.-C. (2011).
\newblock {Detecting Genuine Multipartite Quantum Nonlocality: A Simple
  Approach and Generalization to Arbitrary Dimensions}.
\newblock {\em Physical Review Letters}, 106(2):4.

\bibitem[Bell, 1964]{Bell64}
Bell, J.~S. (1964).
\newblock On the {Einstein-Podolsky-Rosen} paradox.
\newblock {\em Physics}, I:195--200.

\bibitem[Bell, 1966]{Bell1966}
Bell, J.~S. (1966).
\newblock On the problem of hidden variables in quantum mechanics.
\newblock {\em Rev. Mod. Phys.}, 38(3):447--452.

\bibitem[Bj{\"o}rner et~al., 1999]{Bjorner:1999uq}
Bj{\"o}rner, A., Las~Vergnas, M., Sturmfels, B., White, N., and Ziegler, G.
  (1999).
\newblock {\em Oriented matroids}, volume~46 of {\em Encyclopedia of
  Mathematics and Its Applications}.
\newblock Cambridge University Press, 2nd edition.

\bibitem[Butterfield, 1992]{Butterfield92}
Butterfield, J. (1992).
\newblock Bell's theorem: {What it takes}.
\newblock {\em British Journal for the Philosophy of Science}, 43(1):41--83.

\bibitem[Clauser and Horne, 1974]{ClauserHorne1974}
Clauser, J.~F. and Horne, M.~A. (1974).
\newblock Experimental consequences of objective local theories.
\newblock {\em Phys. Rev. D}, 10(2):526--535.

\bibitem[Collins and Gisin, 2004]{CollinsGisin:2004}
Collins, D. and Gisin, N. (2004).
\newblock {A relevant two qubit Bell inequality inequivalent to the CHSH
  inequality}.
\newblock {\em Journal of Physics A: Mathematical and General},
  37(5):1775--1787.

\bibitem[Collins et~al., 2002]{CollinsGisinEtAl:NBodySeparability:2002}
Collins, D., Gisin, N., Popescu, S., Roberts, D., and Scarani, V. (2002).
\newblock {B}ell-type inequalities to detect true n-body nonseparability.
\newblock {\em Physical Review Letters}, 88(17).

\bibitem[Filipp and Svozil, 2005]{FilippSvozil:2005}
Filipp, S. and Svozil, K. (2005).
\newblock {Tracing the bounds on Bell-type inequalities}.
\newblock In {\em AIP Conference Proceedings: Foundations of Probability and
  Physics - 3}, pages 87--94. AIP.

\bibitem[Fine, 1982]{Fine1982}
Fine, A. (1982).
\newblock Hidden variables, joint probability, and the {B}ell inequalities.
\newblock {\em Physical Review Letters}, 48(5):291--295.

\bibitem[Fogel, 2011]{Fogel2011}
Fogel, B. (2011).
\newblock Multiple-context event spaces and distributions: A new framework for
  {B}ell's theorems.
\newblock In progress.

\bibitem[Garg and Mermin, 1982]{GargMermin:1982}
Garg, A. and Mermin, N. (1982).
\newblock {Correlation Inequalities and Hidden Variables}.
\newblock {\em Physical Revew Letters}, 49:1220--1223.

\bibitem[Greenberger et~al., 1990]{GHZS:1990}
Greenberger, D., Horne, M., Shimony, A., and Zeilinger, A. (1990).
\newblock {B}ell's theorem without inequailities.
\newblock {\em American Journal of Physics}, 58(12):1131--1143.

\bibitem[Janotta et~al., 2011]{Janotta:2011by}
Janotta, P., Gogolin, C., Barrett, J., and Brunner, N. (2011).
\newblock {Limits on nonlocal correlations from the structure of the local
  state space}.
\newblock {\em New Journal of Physics}, 13(6):063024.

\bibitem[Mermin, 1990]{Mermin:1990kc}
Mermin, N. (1990).
\newblock {Simple unified form for the major no-hidden-variables theorems}.
\newblock {\em Physical Review Letters}, 65(27):3373--3376.

\bibitem[Muller and Placek, 2001]{MullerPlacek:2001Synthese}
Muller, T. and Placek, T. (2001).
\newblock Against a minimalist reading of {B}ell's theorem: Lessons from
  {F}ine.
\newblock {\em Synthese}, 128(3):343--379.

\bibitem[Oxley, 2011]{Oxley:2011fk}
Oxley, J.~G. (2011).
\newblock {\em Matroid theory}, volume~21 of {\em Oxford graduate texts in
  mathematics}.
\newblock Oxford University Press, Oxford, 2nd ed edition.

\bibitem[Peres, 1999]{Peres99}
Peres, A. (1999).
\newblock All the {B}ell inequalities.
\newblock {\em Foundations of Physics}, 29(4):589--614.

\bibitem[Pironio, 2005]{Pironio2005}
Pironio, S. (2005).
\newblock Lifting {B}ell inequalities.
\newblock {\em Journal of Mathematical Physics}, 46(6):062112.

\bibitem[Pitowsky, 1989]{Pitowsky1989}
Pitowsky (1989).
\newblock Quantum probability -- quantum logic.
\newblock Springer-Verlag.

\bibitem[Pitowsky, 1991]{Pitowsky:1991}
Pitowsky, I. (1991).
\newblock Correlation polytopes: Their geometry and complexity.
\newblock {\em Mathematical Programming}, 50:395--414.

\bibitem[Pitowsky and Svozil, 2001]{PitwoskySvozil2001}
Pitowsky, I. and Svozil, K. (2001).
\newblock Optimal tests of quantum nonlocality.
\newblock {\em Phys. Rev. A}, (014102):4.

\bibitem[Svetlichny, 1987]{Svetlichny:1987ui}
Svetlichny, G. (1987).
\newblock {Distinguishing three-body from two-body nonseparability by a
  Bell-type inequality}.
\newblock {\em Physical Review D}, 35(10):3066--3069.

\bibitem[Svetlichny et~al., 1988]{SvetlichnyEtAl1988}
Svetlichny, G., Redhead, M., Brown, H., and Butterfield, J. (1988).
\newblock Do the {B}ell inequalities require the existence of joint probability
  distributions?
\newblock {\em Philosophy of Science}, 55(3):387.

\bibitem[Uffink, 2002]{Uffink:2002fk}
Uffink, J. (2002).
\newblock Quadratic {B}ell inequalities as tests for multipartite entanglement.
\newblock {\em Physical Review Letters}, 88(230406):4.

\bibitem[Werner and Wolf, 2001]{WernerWolf2001}
Werner, R.~F. and Wolf, M.~M. (2001).
\newblock All-multipartite {B}ell-correlation inequalities for two dichotomic
  observables per site.
\newblock {\em Phys. Rev. A}, 64(3):032112.

\bibitem[{\.Z}ukowski and Brukner, 2002]{ZukowskiBrukner:2002}
{\.Z}ukowski, M. and Brukner, {\v C}. (2002).
\newblock {B}ell's theorem for general n-qubit states.
\newblock {\em Physical Review Letters}, 88(210401):4.

\end{thebibliography}

\appendix

\section{Proof that pioneer sets define tight Bell inequalities\label{sec:Proof-that-pioneer-satisfy-the-TBIC}}

\subsection{Difference function}

The following is a straightforward consequence of \eqref{eq:Pioneer-set-In-out}:
\begin{equation}
\left(\mu,\nu\right)\in\pioneerset_{z}\xor\left(\mu,\nu\cminus\left\{ i\right\} \right)\in\pioneerset_{z}\label{eq:MBP}
\end{equation}
for any $i\in z$ and $\mu,\nu\seq z$.

This allows the definition of a kind of discrete differential:
\begin{equation}
\discretediff{\mu}{\nu}i\equiv\begin{cases}
\emptyset & \left(\mu,\nu\right)\in\pioneerset_{z}\leftrightarrow\left(\mu\cminus\left\{ i\right\} ,\nu\right)\in\pioneerset_{z}\\
\left\{ i\right\}  & \left(\mu,\nu\right)\in\pioneerset_{z}\leftrightarrow\left(\mu\cminus\left\{ i\right\} ,\nu\cminus\left\{ i\right\} \right)\in\pioneerset_{z}
\end{cases}
\end{equation}
Thus,
\begin{equation}
\left(\mu,\nu\right)\in\pioneerset_{z}\leftrightarrow\left(\mu\cminus\left\{ i\right\} ,\nu\cminus\discretediff{\mu}{\nu}i\right)\in\pioneerset_{z}\label{eq:MBP-in terms of the diff function}
\end{equation}

By \eqref{eq:Pioneer-set-In-out}, it can be shown that
\begin{equation}
\discretediff{\mu}{\nu}i\equiv\begin{cases}
\emptyset & \mu\in S_{z}^{*}\leftrightarrow\mu\cminus\left\{ i\right\} \in S_{z}^{*}\\
\left\{ i\right\}  & \mu\in S_{z}^{*}\xor\mu\cminus\left\{ i\right\} \in S_{z}^{*}
\end{cases}\label{eq:difference-function-revised}
\end{equation}
Thus, $\discretediff{\mu}{\nu}i=\discretediff{\mu}{}i$. Note, also,
that $\discretediff{\mu}{\none}i=\discretediff{\mu\cminus\left\{ i\right\} }{\none}i$.

Finally, \eqref{eq:LS-equations-simplified} and \eqref{eq:MBP-in terms of the diff function}
combined with 1a of the TBIC mean that
\begin{equation}
\function f{\mu,\nu}=\function f{\mu\cminus\left\{ i\right\} ,\nu\cminus\discretediff{\mu}{\none}i}\label{eq:LS-equations-on-pioneer-sets}
\end{equation}
Thus, what we need to do is show that the difference function $\ddsymbol$
can connect every two elements of each $S_{z}$.

\subsection{Vertical slices}

We will now construct a series of sequences where consecutive elements
are related by \eqref{eq:LS-equations-on-pioneer-sets}. 

Select some $\nu$ such that $\left(\emptyset,\nu\right)\in S_{z}$
and some $\mu\seq V$. Let $\mu_{n}$ be the nth element of $\mu$.
Construct a sequence of $\left|\mu\right|$ elements according to
the following:
\begin{align}
\left(a_{0},b_{0}\right) & =\left(\emptyset,\nu\right)\\
\left(a_{n},b_{n}\right) & =\left(a_{n-1}\cminus\left\{ \mu_{n}\right\} ,b_{n-1}\cminus\discretediff{a_{n-1}}{\none}{\mu_{n}}\right)
\end{align}
Then $\function f{a_{n},b_{n}}=\function f{a_{n-1},b_{n-1}}=\function f{\emptyset,\nu}$
for all $n$.

By maintaining a constant ordering of the elements of $z$, it is
simple to show that all elements of $\pioneerset_{z}$ are partitioned
according to which element $\left(\emptyset,\nu\right)$ they can
be connected to in the above manner. Thus, we now need only show that
the elements $\left(\emptyset,\text{\ensuremath{\nu}}\right)$ can
be connected to one another.

\subsection{Horizontal slices}

We now want to show that, given any $i,j\in z$, 
\begin{equation}
\function f{\emptyset,\nu}=\function f{\emptyset,\nu\cminus\left\{ i,j\right\} }
\end{equation}
We will do so by constructing a sequence connecting the two elements. 

Select some $Y\in S_{z}$ such that $\left\{ i,j\right\} \seq Y$.
By \#3 in the definition of the pioneer set (see section \ref{sub:Pioneer-set-Definition}),
such a $Y$ must exist. Now take some $\sigma\seq Y\backslash\left\{ i,j\right\} $
and choose an arbitrary order. Let $\sigma_{n}$ be the nth element
of $\sigma$, and let $\sigma_{\leq n}$ be the first $n$ elements
of $\sigma$. 

We now construct a sequence in three parts. Let $\left(a_{0},b_{0}\right)=\left(\emptyset,\nu\right)$.

For $1\leq n\leq\left|\sigma\right|$, the sequence is similar to
that used above:
\begin{align}
\left(a_{n},b_{n}\right) & =\left(a_{n-1}\cminus\left\{ \sigma_{n}\right\} ,b_{n-1}\cminus\discretediff{a_{n-1}}{\none}{\sigma_{n}}\right)
\end{align}

The next 4 elements are:
\begin{align}
\left(a_{\left|\sigma\right|+1},b_{\left|\sigma\right|+1}\right) & =\left(\sigma\cminus\left\{ i\right\} ,\;\;\;\: b_{\left|\sigma\right|}\cminus\discretediff{\sigma}{\none}i\right)\\
\left(a_{\left|\sigma\right|+2},b_{\left|\sigma\right|+2}\right) & =\left(\sigma\cminus\left\{ i,j\right\} ,b_{\left|\sigma\right|+1}\cminus\discretediff{\sigma\cminus\left\{ i\right\} }{\none}j\right)\\
\left(a_{\left|\sigma\right|+3},b_{\left|\sigma\right|+3}\right) & =\left(\sigma\cminus\left\{ j\right\} ,\;\;\; b_{\left|\sigma\right|+2}\cminus\discretediff{\sigma\cminus\left\{ i,j\right\} }{\none}i\right)\\
\left(a_{\left|\sigma\right|+4},b_{\left|\sigma\right|+4}\right) & =\left(\sigma,\hspace{20bp}\;\;\;\;\;\;\;\, b_{\left|\sigma\right|+3}\cminus\discretediff{\sigma\cminus\left\{ j\right\} }{\none}j\right)
\end{align}

The final part of the sequence, another $\left|\sigma\right|$ elements
($\left|\sigma\right|+5\leq n\leq2\left|\sigma\right|+4$), is the
reverse of the first part:
\begin{align}
\left(a_{n},b_{n}\right) & =\left(a_{n-1}\cminus\left\{ \sigma_{\left|\sigma\right|-\left(n-\left(\left|\sigma\right|+5\right)\right)}\right\} ,b_{n-1}\cminus\discretediff{a_{n-1}}{\none}{\sigma_{\left|\sigma\right|-\left(n-\left(\left|\sigma\right|+5\right)\right)}}\right)
\end{align}
So, $a_{2\left|\sigma\right|+4}=\emptyset$. As before, $\function f{a_{n},b_{n}}=\function f{a_{n-1},b_{n-1}}=\function f{\emptyset,\nu}$
for all $n$.

Recall that $\discretediff{\mu}{\none}i=\discretediff{\mu\cminus\left\{ i\right\} }{\none}i$.
This means that 
\begin{equation}
b_{n-1}\cminus b_{n}=b_{\left(2\left|\sigma\right|+4\right)-n}\cminus b_{\left(2\left|\sigma\right|+5\right)-n}
\end{equation}
In other words, the first and last parts of the sequence cancel each
other out, so that
\begin{equation}
b_{\left(2\left|\sigma\right|+4\right)-n}=\nu\cminus\left(b_{\left|\sigma\right|}\cminus b_{\left|\sigma\right|+4}\right)
\end{equation}
and 
\begin{equation}
\left(b_{\left|\sigma\right|}\cminus b_{\left|\sigma\right|+4}\right)=\discretediff{\sigma}{\none}i\cminus\discretediff{\sigma\cminus\left\{ i\right\} }{\none}j\cminus\discretediff{\sigma\cminus\left\{ i,j\right\} }{\none}i\cminus\discretediff{\sigma\cminus\left\{ j\right\} }{\none}j
\end{equation}
Let $\tau_{\sigma}\equiv\left(b_{\left|\sigma\right|}\cminus b_{\left|\sigma\right|+4}\right)$.
Then,
\begin{align}
i\in\tau_{\sigma} & \longleftrightarrow\left(i\in\discretediff{\sigma}{\none}i\xor i\in\discretediff{\sigma\cminus\left\{ i,j\right\} }{\none}i\right)\\
 & \longleftrightarrow\left(\sigma\in S_{z}^{*}\xor\sigma\cminus\left\{ i\right\} \in S_{z}^{*}\xor\sigma\cminus\left\{ j\right\} \in S_{z}^{*}\xor\sigma\cminus\left\{ i,j\right\} \in S_{z}^{*}\right)
\end{align}

Finally, we create a much longer sequence by joining together such
sequences for all $\sigma\seq Y\backslash\left\{ i,j\right\} $. The
final element in this large sequence will be $\left(\emptyset,\nu'\right)$,
where
\begin{equation}
\nu^{'}=\mathop{\cminus}_{\sigma\seq Y\backslash\left\{ i,j\right\} }\tau_{\sigma}
\end{equation}
and
\begin{align}
i\in\nu^{'} & \longleftrightarrow\mathop{\xor}_{\rho\seq Y}\rho\in S_{z}^{*}\\
 & \longleftrightarrow Y\in S_{z}
\end{align}
Since $Y\in S_{z}$ by assumption, $i\in\nu^{'}$. By symmetry, $j\in\nu^{'}$.
Thus,
\begin{equation}
\function f{\emptyset,\nu}=\function f{\emptyset,\nu\cminus\left\{ i,j\right\} }
\end{equation}
which is what we set out to show.

Since $i,j$ were chosen arbitrarily, then for any $\mu,\nu,\rho\seq z$,
where $\left|\rho\right|$ is even,
\begin{equation}
\function f{\emptyset,\nu}=\function f{\mu,\nu\cminus\rho}
\end{equation}

\subsection{Final steps}

The reasoning in the previous subsections can be repeated to show
that, for any $\sigma,\rho,\tau\seq V$, where $\left|\tau\right|$
is even,
\begin{equation}
\function f{\emptyset,\rho}=\function f{\sigma,\rho\cminus\tau}
\end{equation}
Thus, there is some scalar $c$, such that for all $\left(\sigma,\rho\right)\in\pioneerset$,
\begin{equation}
\function f{\sigma,\rho}=c
\end{equation}
The TBIC are thus satisfied. A nonzero solution exists in which all
values are positive ($c>0$), and if any elements are removed from
$\Gamma$, then no non-zero solution exists (because $c$ would have
to equal 0). QED.

\section{Geometric epilogue\label{sec:Geometry and multideviations}}

\subsection{Ordinary probability distributions}

Any probability distribution can be represented as a vector in a vector
space; this allows the set of distributions to be evaluated geometrically.
The unconstrained set of ordinary probability distributions over a
finite, discrete set of cardinality $n$ forms an especially simple
shape, known as a simplex, in $\dsR^{n}$. The multiple-context distributions
of interest in Bell's Theorem form a more complicated polytope, one
whose shape encodes the information in the Bell-type inequalities.
In this section, I will describe how to employ multideviations in
the geometric description of probability distributions.

Given a product set $\Pi A_{B}$, there is a one-to-one correspondence
between real-valued functions over that set and vectors in $\dsR^{n_{B}}$,
where $n_{\sigma}$ is the cardinality function: 
\begin{gather}
\vec{f}=\sum_{\itpx_{B}}f(x)\ \widehat{e}_{\itpx}\\
f(x)=\vec{f}\cdot\widehat{e}_{x}
\end{gather}
where $\{\widehat{e}_{\itpx_{B}}\}_{\itpx_{B}}$ is an orthornormal
basis.

The vectors corresponding to the set of probability distributions
over $\Pi A_{B}$ form a subset of $\dsR^{n_{B}}$, an $\left(n_{B}-1\right)$-dimensional
simplex, the simplest kind of convex polytope (the regular polygons
and regular solids are low-dimensional convex polytopes; triangles
and tetrahedra are examples of simplexes). Each vertex of the simplex
corresponds to a distribution in which one outcome has a probability
of 1, and each facet corresponds to an inequality requiring a particular
probabilty to be greater than zero. Restrictions on the probability
distribution will produce polytopes with more complex shapes.

\subsection{Multidevation vectors}

The MSFs can be used to define multideviation vectors (MD-vectors),
which carve up the vector space into orthogonal subspaces:
\begin{equation}
\MDvector{\sigma}{}{\itpx_{\sigma}}\equiv\sum_{\itpy_{B}}\msf{\sigma}{\none}{\itpx_{\sigma},\itpy_{\sigma}}\hat{e}_{\itpy}
\end{equation}
The MD-vectors decompose the basis vectors: 
\begin{equation}
\hat{e}_{\itpx}=\sum_{\sigma\in\powerset B}\MDvector{\sigma}{}{\itpx_{\sigma}}
\end{equation}
MD-vectors of different order are orthogonal: 
\begin{equation}
\MDvector{\sigma}{}{\itpx_{\sigma}}\cdot\MDvector{\mu}{}{\itpy_{\mu}}=\delta_{\sigma=\mu}\msf{\sigma}{\none}{\itpx_{\sigma},\itpy_{\sigma}}
\end{equation}
The MD-vectors thus carve up the vector space into orthogonal subspaces.
Within each $\sigma$-subspace, the MD-vectors are \emph{not} orthogonal;
indeed, they are not even linearly independent. This is apparent in
the summation property, inherited from the MSFs: 
\begin{equation}
\forall_{i\in\sigma}\left[\sum_{x_{i}}\MDvector{\sigma}{}{\itpx_{\sigma}}=0\right]
\end{equation}

On the other hand, the set of MD-vectors for a given $\sigma$ does
span the $\sigma$-subspace, and the set of $\sigma$-subspaces spans
all of $\dsR^{n_{B}}$. The MD-vectors thus operate as a pseudo-basis;
they can be used to decompose an arbitrary vector, and their inner
product with that vector gives the size of the component: 
\begin{equation}
\vec{v}=\sum_{\sigma\in\powerset B}\sum_{\itpx_{B}}\left(\vec{v}\cdot\MDvector{\sigma}{}{\itpx_{\sigma}}\right)\MDvector{\sigma}{}{\itpx_{\sigma}}\label{Arbitrary vector in MD-representation}
\end{equation}

MD-vectors map functions to vectors:
\begin{equation}
\vec{f}\cdot\MDvector{\sigma}{}{\itpx_{\sigma}}=\multideviation{\sigma}f{\itpx_{\sigma}}{\none}
\end{equation}

For an ordinary probability distribution, then, 
\begin{equation}
\vec{P}=\sum_{\sigma\in\powerset B}\sum_{\itpx_{B}}\MDvector{\sigma}{}{\itpx_{\sigma}}\ \multideviation{\sigma}P{\itpx_{\sigma}}{\none}
\end{equation}
As noted above, the set of all such vectors is an $\left(n_{B}-1\right)$-dimensional
simplex.

\subsection{Multideviation polytopes}

Because the multideviations are segregated into orthogonal subspaces
by their order, we can use them to specify projections into a large
group of subspaces. Such projections will generate new polytopes from
the fundamental simplex. For every $\Psi\seq\powerset B$, the \emph{multideviation
polytope} given by
\begin{equation}
\vec{P}_{\Psi}=\sum_{\sigma\in\Psi}\sum_{\itpx_{B}}\MDvector{\sigma}{}{\itpx_{\sigma}}\ \multideviation{\sigma}P{\itpx_{\sigma}}{\none}
\end{equation}
Each polytope corresponds to the set of distributions that are possible
when certain correlation degrees of freedom are considered accessible.
Finding the facet structure of an arbitrary multideviation polytope
is likely to be an NP-hard problem.

As the projection theorem of section \ref{sec:A-projection-theorem}
shows, the Bell polytopes are a subclass of multideviation polytopes.
In particular, the Bell polytope is a multideviation polytope where
only correlations corresponding to simultaneously realizable joint
measurements are considered accessible.

\subsection{Multiple-context distributions}

For a geometric characterization of multiple-context distributions,
a richer vector space structure is required. There needs to be a separate
vector space for each joint measurement context $\itpp\in\SampleJMS$:
\begin{equation}
\hat{e}_{\itp x}^{\itp p}\cdot\hat{e}_{\itp y}^{\itp q}=\delta_{\itp p=\itp q}\ \delta_{\itp x=\itp y}
\end{equation}
The vectors and distributions are related in a straightforward fashion:
\begin{equation}
\vec{P}=\sum_{\itp p}\sum_{\itp x}P_{\itp p}(\itp x)\ \hat{e}_{\itp x}^{\itp p}\label{MCDistToVec-1}
\end{equation}
 The distribution is recovered from the vector in an equally straightforward
way: 
\begin{equation}
P_{\itp p}(\itp x)=\vec{P}\cdot\hat{e}_{\itp x}^{\itp p}\label{MCVecToDist-1}
\end{equation}

The set of all such vectors describes a convex polytope, although
not a simplex.%
\footnote{It is, rather, the geometric product of a simplex for each joint measurement
context.%
} Once constraints are added, a polytope with a more complicated shape
arises, and the task of finding the tight Bell inequalities is equivalent
to that of finding the facets of the polytope. Generating a description
of a polytope in terms of its facets (known as the $\mathcal{H}$-representation)
given a description in terms of its vertices (the $\mathcal{V}$-representation)
is known as the convex hull problem. For arbitrary polytopes, the
hull problem is known to be NP-complete (see \citealt{Pitowsky:1991}).

\end{document}